\documentclass[aps,twocolumn,floatfix,superscriptaddress,longbibliography]{revtex4}
\usepackage[utf8]{inputenc}
\usepackage[T1]{fontenc}
\usepackage{bm}
\usepackage{times}

\usepackage{algorithm,algorithmic}
\usepackage{mathtools}
\usepackage{amsmath}
\usepackage[shortlabels]{enumitem}
\usepackage{graphicx,epic,eepic,epsfig,amsmath,latexsym,amssymb,verbatim,color}
 
\usepackage{amsfonts}       
\usepackage{nicefrac}       

\usepackage{amsmath}
\usepackage{bbm}

\usepackage{float}
\usepackage{tikz}
\usetikzlibrary{chains}
\usetikzlibrary{fit}
\usepackage{pgflibraryarrows}       
\usepackage{pgflibrarysnakes}       

\usepackage{epsfig}
\usetikzlibrary{shapes.symbols,patterns} 
\usepackage{pgfplots}

\usepackage[strict]{changepage}
\usepackage{hyperref}
\hypersetup{colorlinks=true,citecolor=blue,linkcolor=blue,filecolor=blue,urlcolor=blue,breaklinks=true}

\usepackage[marginal]{footmisc}
\usepackage{url}
\usepackage{theorem}

\newtheorem{proposition}{Proposition}

\def\squareforqed{\hbox{\rlap{$\sqcap$}$\sqcup$}}
\def\qed{\ifmmode\squareforqed\else{\unskip\nobreak\hfil
\penalty50\hskip1em\null\nobreak\hfil\squareforqed
\parfillskip=0pt\finalhyphendemerits=0\endgraf}\fi}
\def\endenv{\ifmmode\;\else{\unskip\nobreak\hfil
\penalty50\hskip1em\null\nobreak\hfil\;
\parfillskip=0pt\finalhyphendemerits=0\endgraf}\fi}
\newenvironment{proof}{\noindent \textbf{{Proof~} }}{\hfill $\blacksquare$}

\newcounter{remark}

\newcounter{example}

\mathchardef\ordinarycolon\mathcode`\:
\mathcode`\:=\string"8000
\def\vcentcolon{\mathrel{\mathop\ordinarycolon}}
\begingroup \catcode`\:=\active
  \lowercase{\endgroup
  \let :\vcentcolon
  }

\usepackage{cleveref}
\usepackage{graphicx}
\usepackage{xcolor}

\RequirePackage[framemethod=default]{mdframed}
\newmdenv[skipabove=7pt,
skipbelow=7pt,
backgroundcolor=darkblue!15,
innerleftmargin=5pt,
innerrightmargin=5pt,
innertopmargin=5pt,
leftmargin=0cm,
rightmargin=0cm,
innerbottommargin=5pt,
linewidth=1pt]{tBox}

\newmdenv[skipabove=7pt,
skipbelow=7pt,
backgroundcolor=blue2!25,
innerleftmargin=5pt,
innerrightmargin=5pt,
innertopmargin=5pt,
leftmargin=0cm,
rightmargin=0cm,
innerbottommargin=5pt,
linewidth=1pt]{dBox}
\newmdenv[skipabove=7pt,
skipbelow=7pt,
backgroundcolor=darkkblue!15,
innerleftmargin=5pt,
innerrightmargin=5pt,
innertopmargin=5pt,
leftmargin=0cm,
rightmargin=0cm,
innerbottommargin=5pt,
linewidth=1pt]{sBox}
\definecolor{darkblue}{RGB}{0,76,156}
\definecolor{darkkblue}{RGB}{0,0,153}
\definecolor{blue2}{RGB}{102,178,255}
\definecolor{darkred}{RGB}{195,0,0}

\newcommand{\nc}{\newcommand}
\nc{\rnc}{\renewcommand}
\nc{\beg}{\begin{equation}}
\nc{\eeq}{{\end{equation}}}
\nc{\beqa}{\begin{eqnarray}}
\nc{\eeqa}{\end{eqnarray}}
\nc{\lbar}[1]{\overline{#1}}
\nc{\bra}[1]{\langle#1|}
\nc{\ket}[1]{|#1\rangle}
\nc{\ketbra}[2]{|#1\rangle\!\langle#2|}
\nc{\braket}[2]{\langle#1|#2\rangle}

\nc{\proj}[1]{| #1\rangle\!\langle #1 |}
\nc{\avg}[1]{\langle#1\rangle}
\nc{\rank}{\operatorname{Rank}}
\nc{\smfrac}[2]{\mbox{$\frac{#1}{#2}$}}
\nc{\tr}{\operatorname{Tr}}
\nc{\ox}{\otimes}
\nc{\dg}{\dagger}
\nc{\dn}{\downarrow}
\nc{\cA}{{\cal A}}
\nc{\cB}{{\cal B}}
\nc{\cC}{{\cal C}}
\nc{\cD}{{\cal D}}
\nc{\cE}{{\cal E}}
\nc{\cF}{{\cal F}}
\nc{\cG}{{\cal G}}
\nc{\cH}{{\cal H}}
\nc{\cI}{{\cal I}}
\nc{\cJ}{{\cal J}}
\nc{\cK}{{\cal K}}
\nc{\cL}{{\cal L}}
\nc{\cM}{{\cal M}}
\nc{\cN}{{\cal N}}
\nc{\cO}{{\cal O}}
\nc{\cP}{{\cal P}}
\nc{\cQ}{{\cal Q}}
\nc{\cR}{{\cal R}}
\nc{\cS}{{\cal S}}
\nc{\cT}{{\cal T}}
\nc{\cV}{{\cal V}}
\nc{\cX}{{\cal X}}
\nc{\cY}{{\cal Y}}
\nc{\cZ}{{\cal Z}}
\nc{\cW}{{\cal W}}
\nc{\csupp}{{\operatorname{csupp}}}
\nc{\qsupp}{{\operatorname{qsupp}}}
\nc{\var}{{\operatorname{var}}}
\nc{\rar}{\rightarrow}
\nc{\lrar}{\longrightarrow}
\nc{\polylog}{{\operatorname{polylog}}}
\nc{\wt}{{\operatorname{wt}}}
\nc{\av}[1]{{\left\langle {#1} \right\rangle}}
\nc{\supp}{{\operatorname{supp}}}

\nc{\argmin}{{\operatorname{argmin}}}

\def\x{\xi}

\nc{\RR}{{{\mathbb R}}}
\nc{\CC}{{{\mathbb C}}}
\nc{\FF}{{{\mathbb F}}}
\nc{\NN}{{{\mathbb N}}}
\nc{\ZZ}{{{\mathbb Z}}}
\nc{\PP}{{{\mathbb P}}}
\nc{\QQ}{{{\mathbb Q}}}
\nc{\UU}{{{\mathbb U}}}
\nc{\EE}{{{\mathbb E}}}
\nc{\id}{{\operatorname{id}}}

\nc{\CHSH}{{\operatorname{CHSH}}}

\nc{\be}{\begin{equation}}
\nc{\ee}{{\end{equation}}}
\nc{\bea}{\begin{eqnarray}}
\nc{\eea}{\end{eqnarray}}
\nc{\<}{\langle}
\rnc{\>}{\rangle}
\nc{\rU}{\mbox{U}}

\nc{\ob}[1]{#1}

\nc{\SEP}{{\text{\rm SEP}}}
\nc{\NS}{{\text{\rm NS}}}
\nc{\LOCC}{{\text{\rm LOCC}}}
\nc{\PPT}{{\text{\rm PPT}}}
\nc{\EXT}{{\text{\rm EXT}}}
\nc{\Sym}{{\operatorname{Sym}}}

\nc{\ERLO}{{E_{\text{r,LO}}}}
\nc{\ERLOCC}{{E_{\text{r,LOCC}}}}
\nc{\ERPPT}{{E_{\text{r,PPT}}}}
\nc{\ERLOCCinfty}{{E^{\infty}_{\text{r,LOCC}}}}
\nc{\Aram}{{\operatorname{\sf A}}}

\usepackage{tikz}
\usepackage{hyperref}
\hypersetup{colorlinks=true,citecolor=blue,linkcolor=blue,filecolor=blue,urlcolor=blue,breaklinks=true}

\makeatletter
\def\grd@save@target#1{%
  \def\grd@target{#1}}
\def\grd@save@start#1{%
  \def\grd@start{#1}}
\tikzset{
  grid with coordinates/.style={
    to path={%
      \pgfextra{%
        \edef\grd@@target{(\tikztotarget)}%
        \tikz@scan@one@point\grd@save@target\grd@@target\relax
        \edef\grd@@start{(\tikztostart)}%
        \tikz@scan@one@point\grd@save@start\grd@@start\relax
        \draw[minor help lines,magenta] (\tikztostart) grid (\tikztotarget);
        \draw[major help lines] (\tikztostart) grid (\tikztotarget);
        \grd@start
        \pgfmathsetmacro{\grd@xa}{\the\pgf@x/1cm}
        \pgfmathsetmacro{\grd@ya}{\the\pgf@y/1cm}
        \grd@target
        \pgfmathsetmacro{\grd@xb}{\the\pgf@x/1cm}
        \pgfmathsetmacro{\grd@yb}{\the\pgf@y/1cm}
        \pgfmathsetmacro{\grd@xc}{\grd@xa + \pgfkeysvalueof{/tikz/grid with coordinates/major step}}
        \pgfmathsetmacro{\grd@yc}{\grd@ya + \pgfkeysvalueof{/tikz/grid with coordinates/major step}}
        \foreach \x in {\grd@xa,\grd@xc,...,\grd@xb}
        \node[anchor=north] at (\x,\grd@ya) {\pgfmathprintnumber{\x}};
        \foreach \y in {\grd@ya,\grd@yc,...,\grd@yb}
        \node[anchor=east] at (\grd@xa,\y) {\pgfmathprintnumber{\y}};
      }
    }
  },
  minor help lines/.style={
    help lines,
    step=\pgfkeysvalueof{/tikz/grid with coordinates/minor step}
  },
  major help lines/.style={
    help lines,
    line width=\pgfkeysvalueof{/tikz/grid with coordinates/major line width},
    step=\pgfkeysvalueof{/tikz/grid with coordinates/major step}
  },
  grid with coordinates/.cd,
  minor step/.initial=.2,
  major step/.initial=1,
  major line width/.initial=2pt,
}
\makeatother

\usepackage{thmtools}
\usepackage{thm-restate}
\usepackage{etoolbox}
\makeatletter
\def\problem@s{}
\newcounter{problems@cnt}

\newcommand{\allproblems}{\problem@s}
\makeatother

\usepackage{tcolorbox}
\usepackage{relsize}
\usepackage{graphicx}
\usepackage[qm]{qcircuit}
\nc{\st}{\text{subject to} \ }
\nc{\supre}{\text{supremum} \ }
\nc{\sdp}{\text{sdp}}
\usepackage{array}

\allowdisplaybreaks

\begin{document}
\title{Practical distributed quantum information processing with LOCCNet}
\author{Xuanqiang Zhao}
\author{Benchi Zhao}
\author{Zihe Wang}
\author{Zhixin Song}
\author{Xin Wang}
\thanks{wangxin73@baidu.com}
\affiliation{Institute for Quantum Computing, Baidu Research, Beijing 100193, China}

\begin{abstract}
Distributed quantum information processing is essential for building quantum networks and enabling more extensive quantum computations. In this regime, several spatially separated parties share a multipartite quantum system, and the most natural set of operations is Local Operations and Classical Communication (LOCC). As a pivotal part in quantum information theory and practice, LOCC has led to many vital protocols such as quantum teleportation. However, designing practical LOCC protocols is challenging due to LOCC's intractable structure and limitations set by near-term quantum devices. Here we introduce LOCCNet, a machine learning framework facilitating protocol design and optimization for distributed quantum information processing tasks. As applications, we explore various quantum information tasks such as entanglement distillation, quantum state discrimination, and quantum channel simulation. We discover protocols with evident improvements, in particular, for entanglement distillation with quantum states of interest in quantum information. Our approach opens up new opportunities for exploring entanglement and its applications with machine learning, which will potentially sharpen our understanding of the power and limitations of LOCC. An implementation of LOCCNet is available in Paddle Quantum, a quantum machine learning Python package based on PaddlePaddle deep learning platform.
\vspace{1.2cm}
\end{abstract}

\date{\today}
\maketitle

\section{Introduction}
In the past few decades, quantum technologies have been found to have an increasing number of powerful applications in areas including optimization~\cite{Farhi2014,Harrigan2020}, chemistry~\cite{McArdle2018a,Arute2020}, security~\cite{Bennett1984,Ekert1991}, and machine learning~\cite{Biamonte2017a}. To realize large-scale quantum computers and deliver real-world applications, distributed quantum information processing will be essential in the technology road map, where quantum entanglement and its manipulation play a crucial role.

Quantum entanglement is central to quantum information by serving as a fundamental resource which underlies many important protocols such as teleportation~\cite{bennett1993teleporting}, superdense coding~\cite{bennett1992communication}, and quantum cryptography~\cite{Ekert1991}.
To achieve real-world applications of quantum technologies, protocols for manipulating quantum entanglement are essential ingredients, and it will be important to improve existing methods. The study of entanglement manipulation is one of the most active and important areas in quantum information~\cite{Plenio2007, horodecki2009quantum}.

In entanglement manipulation and distributed quantum information processing, multiple spatially separated parties are usually involved. As direct transfers of quantum data between these nodes are not feasible with current technology, Local Operations and Classical Communication (LOCC)~\cite{bennett1993teleporting} is more practical at this stage. Such an LOCC (or distant lab) paradigm plays a fundamental role in entanglement theory, and many important results have been obtained within this paradigm~\cite{horodecki2009quantum}. However, how to design LOCC protocols on near-term quantum devices~\cite{preskill2018quantum} remains an important challenge. Such protocols are generally hard to design even with perfect entanglement due to the complicated and hard-to-characterize structure of LOCC~\cite{Chitambar2014}. Moreover, limited capabilities and structure of near-term quantum devices have to be considered during the design of LOCC protocols. 

Inspired by the breakthroughs of deep learning~\cite{LeCun2015} in mastering the game of Go \cite{Silver2016} and solving  protein folding~\cite{Jumper2020}, it is desirable to apply machine learning ideas to explore quantum technologies. For instance, machine learning has been applied to improve quantum processor designs~\cite{Mavadia2017,Wan2017,Lu2017b,Niu2019} and quantum communication~\cite{Wallnofer2020,Bausch2018}.
Here, we adopt the ideas from machine learning to solve the challenges in exploring LOCC protocols. We use parameterized quantum circuits (PQC)~\cite{Benedetti2019a} to represent the local operations allowed in each spatially separated party and then incorporate multiple rounds of classical communication.
Then one can formulate the original task as an optimization problem and adopt classical optimization methods to search the optimal LOCC protocol. The PQCs have been regarded as machine learning models with remarkable expressive power, which leads to applications in quantum chemistry and optimization~\cite{Benedetti2019a}. Here, we generalize PQC to a larger deep learning network to deal with distributed quantum information processing tasks and in particular to explore better entanglement manipulation protocols.

In this work, we introduce a machine learning framework for designing and optimizing LOCC protocols that are adaptive to near-term quantum devices, which consists of a set of PQCs representing local operations.
As applications, we explore central quantum information tasks such as entanglement distillation, state discrimination, and quantum channel simulation.
We discover protocols with evident improvements via this framework, sharpening our understanding of the power and limitations of LOCC.
As showcases, we establish hardware-efficient and simple protocols for entanglement distillation and state discrimination, which outperforms previously best-known methods. In particular, for distillation of Bell states with non-orthogonal product noise, the optimized protocol outputs a state whose distillation fidelity even reaches the theoretical upper bound and hence is optimal.

\section{Results}
\subsection{The LOCCNet framework}
In this section, we introduce {LOCCNet, a machine learning framework} that facilitates the design of LOCC protocols for various quantum information processing tasks, including entanglement distillation~\cite{Bennett1996,deutsch1996quantum,Murao1998,Dur2007,Pan2003,Devetak2003a}, quantum state discrimination~\cite{Bennett1999b, Walgate2000, Fan2004a, Hayashi2006, Ghosh2004, Nathanson2005, Duan2007a, Chitambar2014b, Duan2009d, Childs2013, Li2017, Bandyopadhyay2014a}, and quantum channel simulation~\cite{Bennett1996c,Bennett2014,Berta2013,Pirandola2015b,Wilde2018,WW18,Fang2018}. An LOCC protocol can be characterized as a sequence of local quantum operations performed by spatially separated parties with classical communication of measurement outcomes [see Supplementary Note 1].

According to the number of classical communication rounds, one can divide LOCC into different classes~\cite{Chitambar2014}. The one-round protocols correspond to LOCC operations where one party applies a local operation and sends the measurement outcome to others, who then apply local operations chosen based on the outcome they receive. Based on one-round protocols, we are able to construct an $r$-round protocol recursively. All these protocols belong to the finite-round LOCC class, and can be visualized as tree graphs. Each node in the tree represents a local operation and different measurement outcomes correspond to edges connecting to this node's children, which represent different choice of local operations based on the measurement outcomes from last round.

Although the basic idea of LOCC is relatively easy to grasp, its mathematical structure is highly complicated~\cite{Chitambar2014} and hard to characterize. As indicated by its tree structure, a general $r$-round LOCC protocol could lead to exponentially many possible results, making LOCC protocol designs for many essential quantum information processing tasks very challenging.
At the same time, it will be more practical to consider LOCC protocols with hardware-efficient local operations and a few communication rounds due to the limited coherence time of local quantum memory. To overcome these challenges, we propose to find LOCC protocols with the aid of machine learning, inspired by its recent success in various areas. Specifically, we present the LOCCNet framework, which incorporates optimization methods from classical machine learning field into the workflow of designing LOCC protocols and can simulate any finite round LOCC in principle. 

\begin{figure}[t]
\centering
\includegraphics[width=0.5\textwidth]{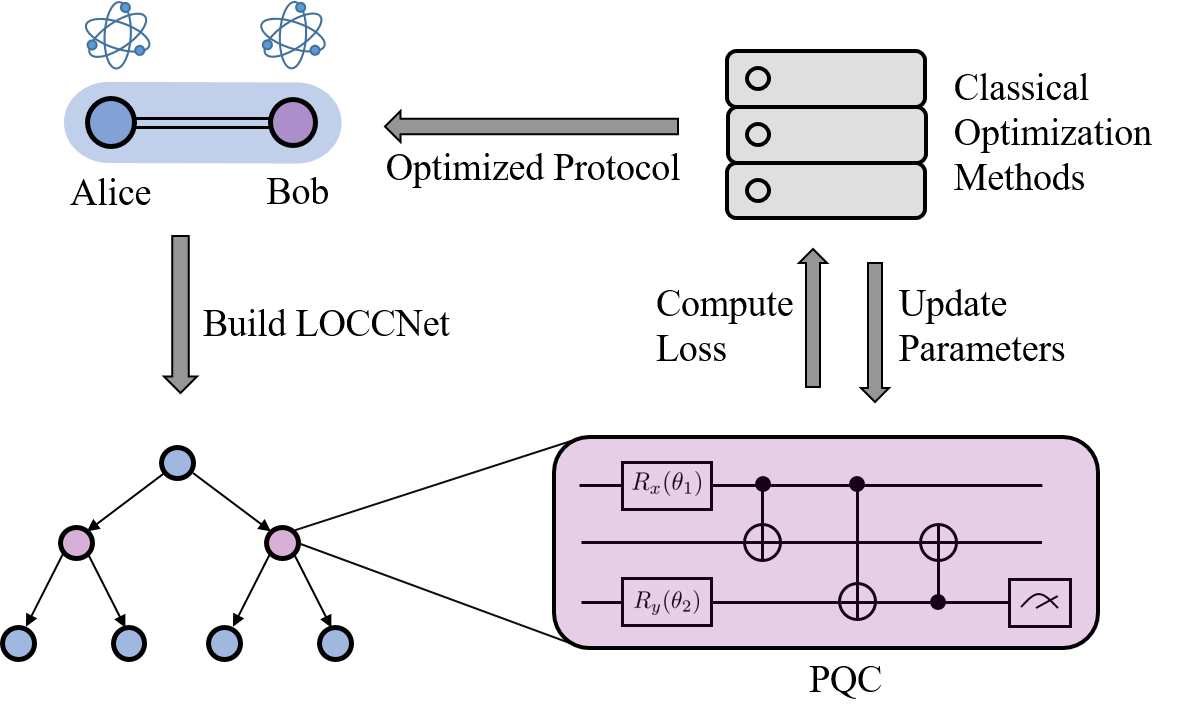}
\caption{\textbf{Illustration of the procedure for optimizing an LOCC protocol with LOCCNet.} For simplicity, only two parties are involved in this workflow, namely Alice and Bob. The tree presented here corresponds to a specific two-round LOCC protocol. Such a tree can be customized with LOCCNet. With each node (Local Operation) encoded as a PQC and arrows between nodes referring to classical communication, one can define a loss function to guide the training process depending on the task. The tree branch diverges indicating different possible measurement outcomes. Finally, one can adopt optimization methods to iteratively update the parameters $\bm{\theta}$ in each local operation and hence obtain the optimized LOCC protocol. 
}
\label{fig:LOCCNet_overview}
\end{figure}

As illustrated in Fig.~\ref{fig:LOCCNet_overview}, each party's local operations, represented by nodes in a tree, are described as parameterized quantum circuits (PQC)~\cite{Benedetti2019a}.
Users can measure any chosen qubit and define a customized loss function from measurement outcomes as well as remaining states. With a defined loss function for a task of interest, LOCCNet can be optimized to give a protocol. The effect of classical communication is also well simulated by LOCCNet in the sense that different PQCs can be built for different measurement outcomes from previous rounds. 

{
Previously, PQCs have been adapted to many research areas including quantum simulation~\cite{peruzzo2014variational}, quantum optimization~\cite{Farhi2014}, and quantum error correction~\cite{Johnson2017}. The family of variational quantum algorithms~\cite{Cerezo2020,Bharti2021,Endo2020}, based on PQCs, is one promising candidate to achieve quantum advantages with near-term devices. In quantum information, PQCs also help in estimating distance measures for quantum states~\cite{Chen2020a,Cerezo2019} and compressing quantum data~\cite{Romero2017,Cao2020}. Here, we take one step further by extending the use of PQCs to the distributed quantum information processing scenario where LOCC is the most natural set of operations.
}

In the next three sections, we will demonstrate the LOCCNet framework in details with important applications and present some interesting findings, including protocols that achieve better results than existing ones. We conduct software implementations of LOCCNet using the \href{https://qml.baidu.com/}{Paddle Quantum} toolkit~\cite{Paddlequantum} on the PaddlePaddle Deep Learning
Platform~\cite{Ma2019,Paddle}.

\subsection{Entanglement distillation}
Many applications of LOCC involve entanglement manipulation, and the use of entanglement is generally required to be in its pure and maximal form. Hence, the efficient conversion of entanglement into such a form, a process known as entanglement distillation~\cite{Bennett1996,Bennett1996c}, is usually a must for many quantum technologies. The development of entanglement distillation methods remains at the forefront of quantum information~\cite{horodecki2009quantum}. For example, the two-qubit maximally entangled state $|\Phi^+\rangle = 1/\sqrt{2}(|00\rangle+|11\rangle)$, which is also known as the entangled bit (ebit), is the fundamental resource unit in entanglement theory since it is a key ingredient in many quantum information processing tasks. Thus, an essential goal for entanglement distillation in a two-qubit setting is to convert a number of copies of some two-qubit state $\rho_{AB}$ shared by two parties, Alice and Bob, into a state as close as possible to the ebit. Here, closeness between the state $\rho_{AB}$ and the ebit is usually measured in terms of the fidelity
\begin{align}
F = \langle\Phi^+|\rho_{AB}|\Phi^+\rangle.
\end{align}

Although theory is more concerned with asymptotic distillation with unlimited copies of $\rho_{AB}$, protocols considering a finite number of copies are more practical due to the physical limitations of near-term quantum technologies. Also, practical distillation protocols usually allow for the possibility of failure as a trade-off for achieving a higher final fidelity. Furthermore, due to limited coherence time of local quantum memories, schemes involving only one round of classical communication are preferred in practice.
Under these settings, many practical schemes for entanglement distillation have been proposed~\cite{Bennett1996, deutsch1996quantum, fujii2009entanglement, kalb2017entanglement, rozpkedek2018optimizing, krastanov2019optimized}. Not surprisingly, there is not a single scheme that applies to all kinds of states. In fact, designing a protocol even for a specific type of states is a difficult task.

In this section, we apply LOCCNet to entanglement distillation and present selected results that reinforce the validity and practicality of using this framework for designing LOCC protocols. To use LOCCNet for finding distillation protocols for a state $\rho_{AB}$, we build two PQCs, one for Alice and one for Bob. In the preset event of success, these PQCs output a state supposed to have a higher fidelity to the ebit. To optimize PQCs, we define the infidelity of the output state and the ebit, i.e., $1 - F$, as the loss function to be minimized. As soon as the value of the loss function converges through training, the PQCs along with the optimized parameters form an LOCC distillation protocol. In principle, this training procedure is general and can be applied to find distillation protocols for any initial state $\rho_{AB}$ given its numerical form. Beyond rediscovering existing protocols, we are also able to find improved protocols with LOCCNet. Below, we give two distillation protocols for S states and isotropic states, respectively, as examples of optimized schemes found with LOCCNet.

An S state is a mixture of the ebit $|\Phi^+ \rangle$ and non-orthogonal product noise~\cite{rozpkedek2018optimizing}.
Here, we define it to be
\begin{align}
\rho_{AB} = p|\Phi^+\rangle\langle\Phi^+| + (1-p)|00\rangle\langle00|,
\end{align}
where $p\in[0,1]$. A distillation protocol known to perform well on two copies of some S state is the DEJMPS protocol~\cite{deutsch1996quantum}, which in this case outputs a state whose fidelity to the ebit is $(1+p)^2/(2+2p^2)$ with a probability of $(1+p^2)/2$ [see Supplementary Note 2].

Here, we present a protocol learned by LOCCNet that can output a state achieving a fidelity higher than DEJMPS and close to the highest possible fidelity. Details on this protocol after simplification are given in Fig.~\ref{fig:s_state_distill_circuit}, where Alice and Bob apply local operations to their own qubits independently and then compare their measurement outcomes through classical communication. The distillation succeeds only when both Alice and Bob get $0$ from computational basis measurements.

\begin{figure}[!hbtp]
\centerline{
\Qcircuit @C=1em @R=1.7em {
\lstick{A_0} & \targ &\qw &\qw &\qw \\
\lstick{A_1} & \ctrl{-1} &\qw & \gate{R_y(\theta)} &\meter \\
\lstick{B_0} & \targ & \ctrl{1} &\qw &\qw \\
\lstick{B_1} & \ctrl{-1} & \targ & \gate{R_y(\theta)} &\meter \\
}
}
\caption{\textbf{Circuit of a distillation protocol learned by LOCCNet for S states.} This simplified circuit represents local operations in a distillation protocol learned by LOCCNet for two copies of an S state, $\rho_{A_0B_0}$ and $\rho_{A_1B_1}$. The rotation angles of both $R_{y}$ gates are $\theta = \arccos(1-p)+\pi$, which depend on the parameter $p$ of the S states to be distilled.}
\label{fig:s_state_distill_circuit}
\end{figure}
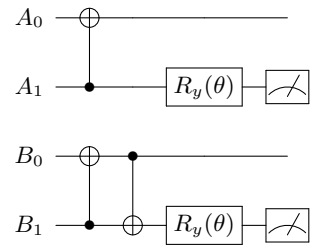

The final fidelity achieved by this protocol is compared with that achieved by the DEJMPS protocol in Fig.~\ref{fig:s_state_fid}. 
For the aim of benchmarking, the techniques based on partial positive transpose (PPT) were introduced to derive fundamental limits of entanglement distillation~\cite{Rains2001,Matthews2008,Wang2016,Fang2017,rozpkedek2018optimizing,Wang2016c}.  The entanglement theory under PPT operations has been extensively studied in the literature~(e.g.,~\cite{Audenaert2003,Wang2016d,Regula2019,Chitambar2017,Wang2017d,Wang2020c}) and offers valuable limitations of LOCC.
Here, the PPT bound obtained with semi-definite programming~\cite{rozpkedek2018optimizing} is an upper bound to the fidelity achieved by any LOCC protocol [see Supplementary Note 2]. 

As shown in the figure, the protocol learned by LOCCNet achieves near-optimal fidelity in the sense that it is close to the PPT bound. Analytically, for two copies of some S state with a parameter $p$, the post-measurement state in the event of success is $\sigma_{AB} = F|\Phi^+\rangle\langle\Phi^+| + (1-F)|\Phi^-\rangle\langle\Phi^-|$, where
\begin{align}
F = \frac{1+\sqrt{2p-p^2}}{2}
\end{align}
is its fidelity to the ebit and $|\Phi^-\rangle = 1/\sqrt{2}(|00\rangle-|11\rangle)$. The probability of arriving at this state is $p_\text{succ} = p^2-p^3/2$ [see Supplementary Note 2]. It is noteworthy that the distilled state is a Bell diagonal state of rank two. For two copies of such a state, the DEJMPS protocol achieves the optimal fidelity~\cite{rozpkedek2018optimizing,Ruan2018}. Thus, combining our protocol with the DEJMPS protocol offers an efficient and scalable distillation scheme for more copies of some S state.

\begin{figure}
\centering
\includegraphics[width=0.5\textwidth]{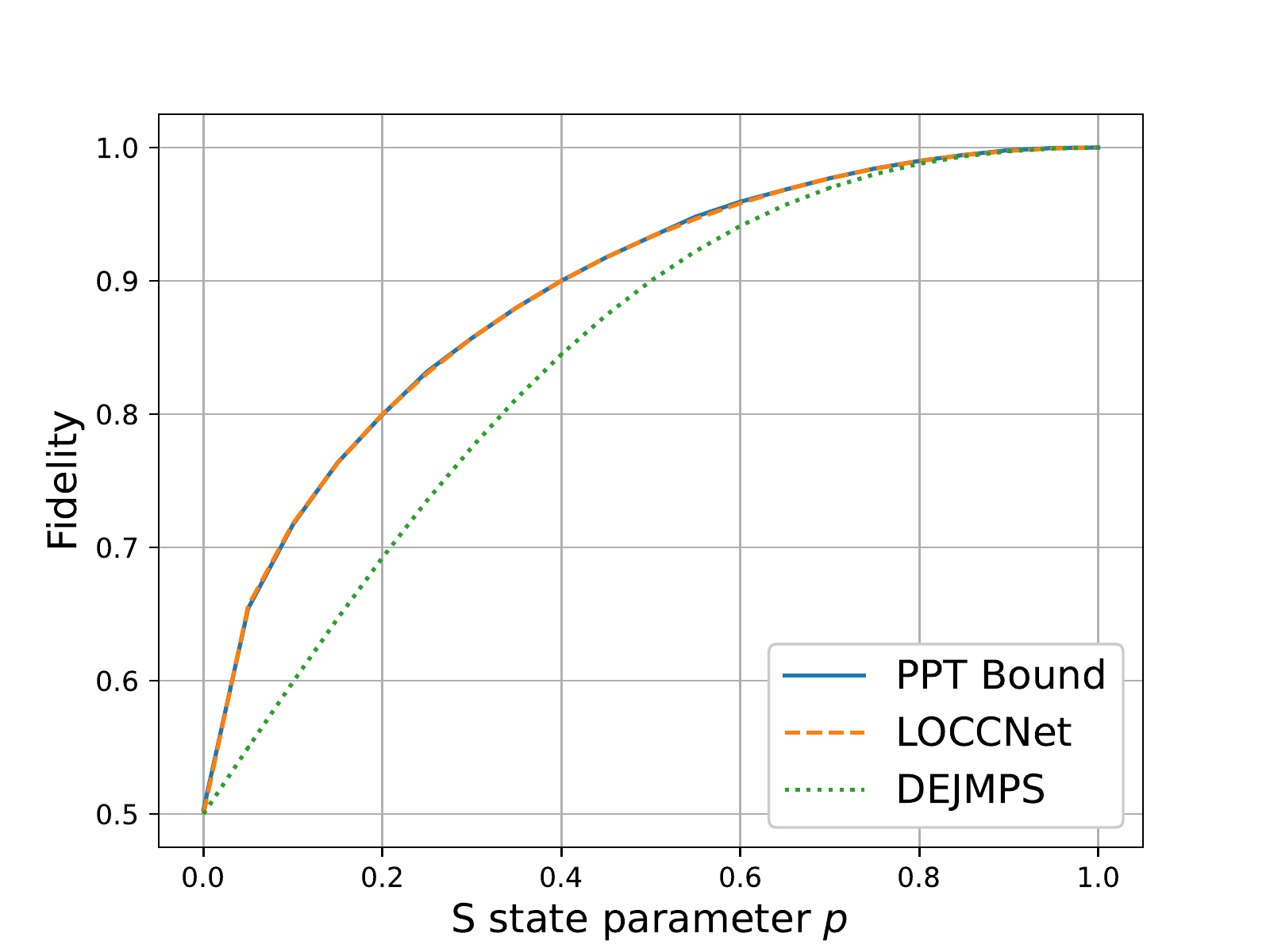}
\caption{\textbf{Fidelity achieved by distillation protocols for two copies of some S state.} The orange dashed line depicts the performance of the protocol learned by LOCCNet, which outperforms the DEJMPS protocol (green dotted). Also, the learned protocol is near optimal in the sense that its line almost aligns with the PPT bound (blue solid).}
\label{fig:s_state_fid}
\end{figure}

Another important family of entangled states is the isotropic state family, defined as
\begin{align}
\rho_{AB} = p|\Phi^+\rangle\langle\Phi^+| + (1-p)\frac{I}{4},
\end{align}
where $p\in[0,1]$ and $I$ is the identity matrix. Distillation protocols for two copies of some isotropic state have been well studied, and the DEJMPS protocol achieves empirically optimal fidelity in this case. Given four copies of some isotropic state with a parameter $p$, a common way to distill entanglement is to divide them into two groups of two copies and apply the DEJMPS protocol to each group. Conditioned on success, we then apply the DEJMPS protocol again to the two resulting states from the previous round. Since the DEJMPS protocol was originally designed for two-copy distillation, such a generalization is probably unable to fully exploit the resources contained in four copies of the state. Indeed, with the aid of LOCCNet, we find a protocol optimized specifically for four copies of some isotropic state. As illustrated in Fig.~\ref{fig:4copy_iso_circuit}, Alice and Bob first apply similar local operations with three pairs of qubits being measured and then compare their measurement outcomes through classical communication. If their measurement outcomes for each pair of qubits are identical, the distillation procedure succeeds.

\begin{figure}[hbtp]
\centerline{
\Qcircuit @C=1em @R=1.7em {
\lstick{A_0} & \ctrl{1} & \qw & \qw & \targ & \qw &\qw \\
\lstick{A_1} & \targ & \ctrl{1} & \qw & \qw & \gate{R_x(+\frac{\pi}{2})} & \meter \\
\lstick{A_2} & \qw & \targ & \ctrl{1} & \qw & \gate{R_x(+\frac{\pi}{2})} & \meter \\
\lstick{A_3} & \qw & \qw & \targ & \ctrl{-3} & \gate{R_x(+\frac{\pi}{2})} & \meter \\
}
}
\caption{\textbf{Circuit of a distillation protocol learned by LOCCNet for isotropic states.} This simplified circuit represents Alice's local operation in a protocol learned by LOCCNet for entanglement distillation with four copies of some isotropic state. Bob's local operation is identical to Alice's, except that the rotation angles of Bob's $R_x$ gates are $-\pi/2$.}
\label{fig:4copy_iso_circuit}
\end{figure}
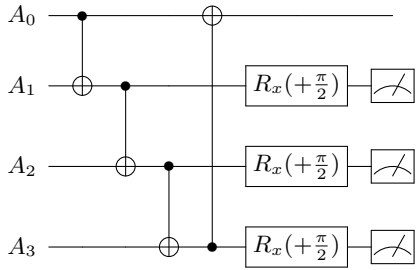

The fidelity achieved by this protocol for different input isotropic states is plotted in Fig.~\ref{fig:iso_state_fid}, along with that of the generalized DEJMPS protocol. For four copies of some isotropic state with a parameter $p$, our protocol achieves a final fidelity of
\begin{align}
F = \frac{1-2p+9p^2}{4-8p+12p^2},
\end{align}
which is slightly higher than the DEJMPS protocol, as shown in Fig.~\ref{fig:iso_state_fid}.
Details are referred to [Supplementary Note 2]. Another advantage of this optimized protocol is that the output state in the event of success is still an isotropic state, implying the possibility of a generalized distillation protocol for $4^n$ copies of some isotropic state.

We remark that our protocols are optimized with the goal to achieve the highest possible fidelity, so their probabilities of success are not high. For situations where the probability of success is important, one can also design a customized loss function to optimize a protocol according to their metrics.

\begin{figure}[!hbtp]
\centering
\includegraphics[width=0.5\textwidth]{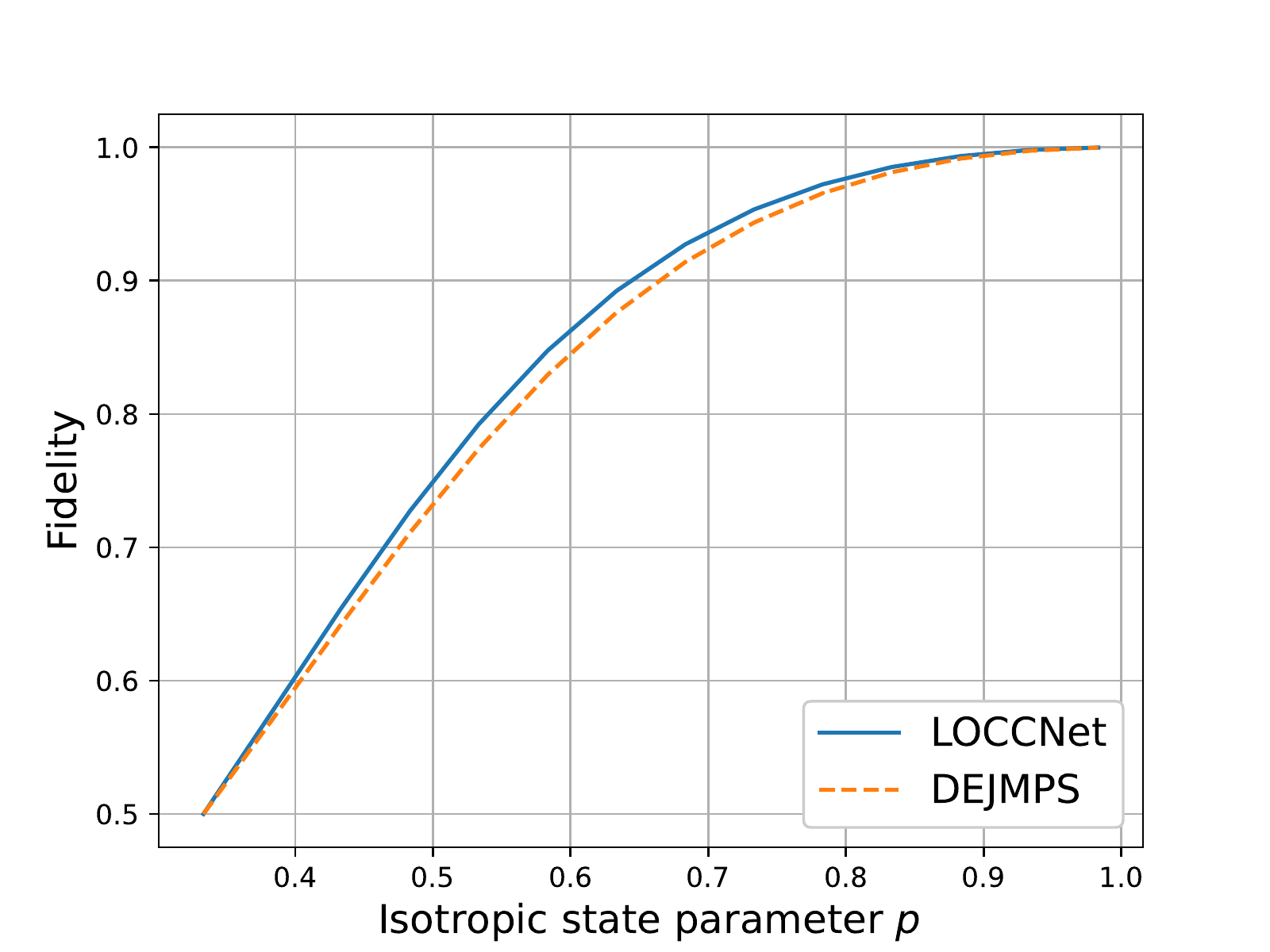}
\caption{\textbf{Fidelity achieved by distillation protocols for four copies of some isotropic state.} The blue solid line depicts the fidelity achieved by the protocol learned by LOCCNet, which outperforms the generalized DEJMPS protocol (orange dashed).}
\label{fig:iso_state_fid}
\end{figure}

\subsection{Distributed quantum state discrimination}
Another important application of LOCC is quantum state discrimination (QSD). Distinguishing one physical configuration from another is central to information theory. When messages are encoded into quantum states for information transmission, the processing of this information relies on the distinguishability of quantum states. Hence, QSD has been a central topic in quantum information~\cite{Bae2015,reviewQSD,Li2015a}, which investigates how well quantum states can be distinguished and underlies various applications in quantum information processing tasks, including quantum data hiding~\cite{DiVincenzo2002} and dimension witness~\cite{Gallego2010}.

QSD using global quantum operations is well-understood in the sense that the optimal strategy maximizing the success probability can be solved efficiently via semi-definite programming (SDP)~\cite{Eldar2003,Sun2002,Jezek2002}. However, for an important operational setting called distant lab paradigm or distributed regime, our knowledge of QSD remains limited despite substantial efforts in the past two decades~\cite{Bennett1999b, Walgate2000, Fan2004a, Hayashi2006, Ghosh2004, Nathanson2005, Duan2007a, Chitambar2014b, Duan2009d, Childs2013, Li2017, Bandyopadhyay2014a}. In the distributed regime, multipartite quantum states are distributed to spatially separated labs, and the goal is to distinguish between these states via LOCC.

For two orthogonal pure states shared between multiple parties, it has been shown that they can be distinguished via LOCC alone no matter if these states are entangled or not~\cite{Walgate2000}. However, it is not easy to design a concrete LOCC protocol for practical implementation on near-term quantum devices. Using LOCCNet, one can optimize and obtain practical LOCC protocols for quantum state discrimination. Furthermore, for non-orthogonal states, limited aspects have been investigated in terms of the feasibility of LOCC discrimination. However, LOCCNet can provide an optimized and practical protocol in this realistic setting.

Here, to explore the power of LOCCNet in state discrimination, we focus on the optimal success probability of discriminating between noiseless and noisy Bell states via LOCC. Consider two Bell states, $|\Phi^+\rangle$ and $|\Phi^-\rangle$, and an amplitude damping (AD) channel $\mathcal{A}$ with noise parameter $\gamma$ such that $\mathcal{A}(\rho) = E_0\rho E_0^\dagger+E_1\rho E_1^\dagger$ with $E_0=|0\rangle\langle0|+\sqrt{1-\gamma}|1\rangle\langle1|$ and $E_1=\sqrt{\gamma}|0\rangle\langle1|$. If we send $|\Phi^-\rangle$'s two qubits respectively through this AD channel, then the resulting state is $\mathcal{A}\otimes\mathcal{A}(|\Phi^-\rangle\langle\Phi^-|)$.
The goal is now to distinguish between $\proj{\Phi^+}$ and $\mathcal{A}\otimes\mathcal{A}(\proj{\Phi^-})$. 

Suppose $\Phi_0$ and $\Phi_1$ are some pair of two-qubit states. To find a protocol discriminating between them, we build an ansatz with measurements on both qubits. As illustrated in Fig.~\ref{fig:QSD_circuit}, Alice performs a unitary gate on her qubit followed by a measurement, whose outcome determines Bob's operation on his qubit.  Given an ideal discrimination protocol, Bob's measurement outcome should be $0$ if and only if the input state is $\Phi_0$ so that he can tell which state the input state is for sure. Based on this observation, we define a loss function
\begin{align}
L = P(1|\Phi_0) + P(0|\Phi_1),
\end{align}
where $P(j|\Phi_k)$ is the probability of Bob's measurement outcome being $j$ given the input state being $\Phi_k$. By minimizing this loss function, we are able to obtain a protocol for distinguishing between states $\Phi_0$ and $\Phi_1$ with an optimized probability of success. Specifically, for $\Phi_0\equiv\proj{\Phi^+}$ and $\Phi_1\equiv\mathcal{A}\otimes\mathcal{A}(\proj{\Phi^-})$, through optimization we find a protocol where Alice's local unitary operation is $U=R_y(\pi/2)$ and Bob's local unitary operation is $V=R_y((-1)^a\theta)$ where $\theta = \pi-\arctan((2-\gamma)/\gamma)$ and $a=0\text{ or }1$ is Alice's measurement outcome. This optimized protocol achieves an average success probability of
\begin{align}
p_\text{succ} = \frac{1}{2} + \frac{\sqrt{2-2\gamma+\gamma^2}}{2\sqrt{2}},
\end{align}
as given in [Supplementary Note 3].

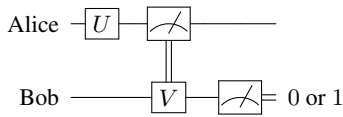
\begin{figure}[!hbtp]
\centerline{
\Qcircuit @C=0.6em @R=1.7em {
\lstick{\text{Alice}} & \gate{U} & \qw & \meter \cwx[1] & \qw & \qw & \qw \\
\lstick{\text{Bob}} & \qw & \qw  & \gate{V} & \qw & \meter & \rstick{0\text{ or }1}\cw \\
}
}
\caption{\textbf{Ansatz used for finding QSD protocols with LOCCNet.} Alice performs a unitary gate on her qubit and measures. Then Bob performs on his qubit a unitary gate chosen based on Alice's measurement result. Bob's measurement outcome is supposed to tell which state the input state is.}
\label{fig:QSD_circuit}
\end{figure}

In Fig.~\ref{fig:QSD_exp}, we compare the protocol learned by LOCCNet with the optimal protocol for perfect discrimination between two noiseless and orthogonal Bell states $|\Phi^+\rangle$ and $|\Phi^-\rangle$. The PPT bound shown in Fig.~\ref{fig:QSD_exp} is obtained via SDP and serves as an upper bound to the average probability of any LOCC protocol recognizing the input state correctly~\cite{Yu2014a}, where the input state is either $\Phi_0$ or $\Phi_1$ with equal chance. While the noiseless protocol is consistently better than random guessing as noise in the AD channel increases, it inevitably suffers from a decrease in its discrimination ability. The gap between its probability of success and the PPT bound steadily widens. On the other hand, the protocol optimized with LOCCNet can achieve a near-optimal probability of success for each noise setting, as shown in the figure.

\begin{figure}[h]
    \centering
    \includegraphics[width=0.5\textwidth]{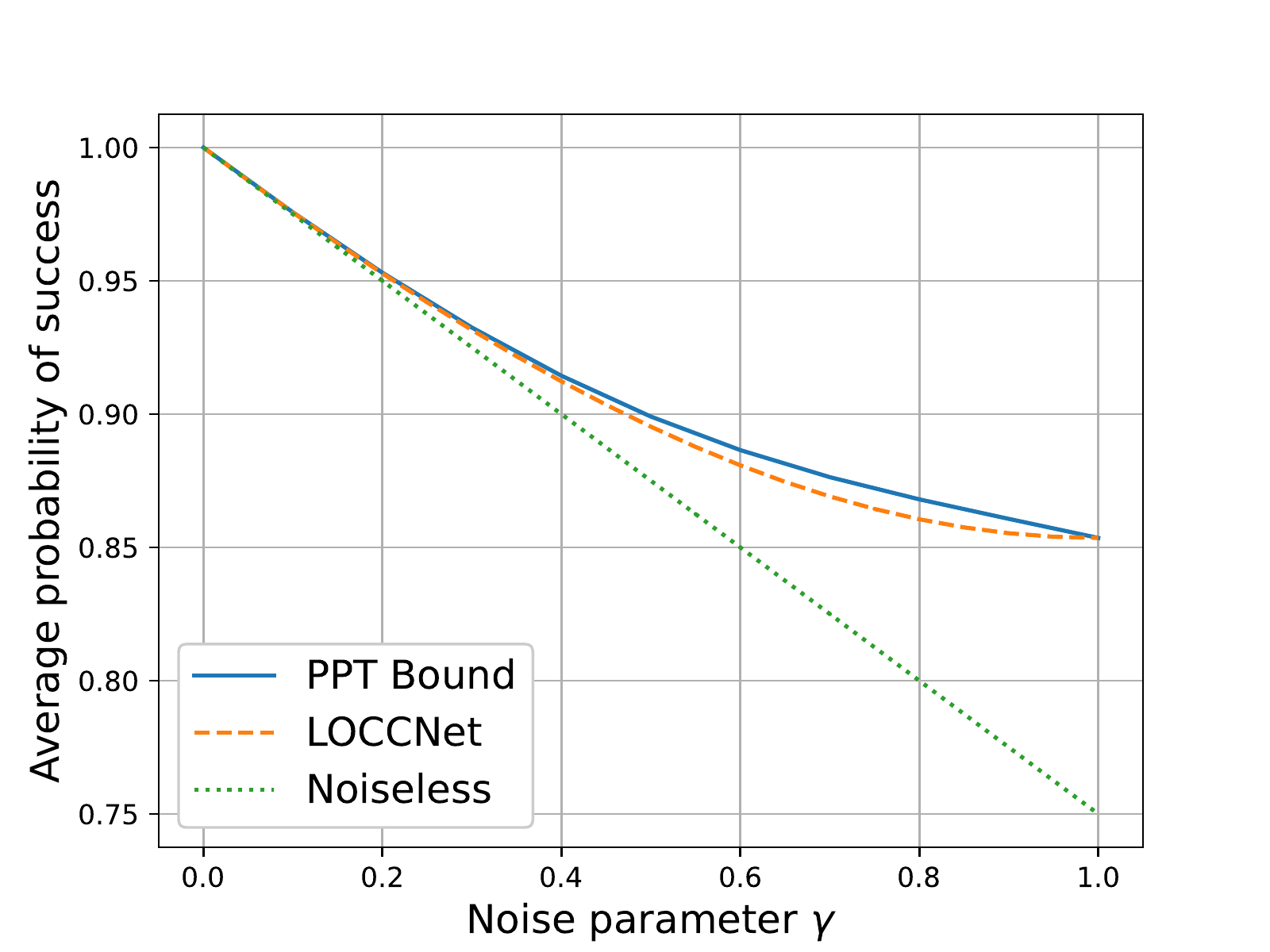}
    \caption{\textbf{Average success probability of distinguishing a Bell state and a noisy Bell state.} The orange dashed line depicts the behavior of the protocol via LOCCNet, which outperforms the protocol for distinguishing perfect orthogonal Bell states (green dotted). Moreover, the protocol from LOCCNet is near optimal since it almost matches the upper bounds obtained via PPT POVMs~(blue solid).}
    \label{fig:QSD_exp}
\end{figure}

\subsection{Quantum channel simulation}
One central goal of quantum information is to understand the limitations governing the use of quantum systems to take advantage of quantum physics laws. Quantum channel lies at the heart of this question since it characterizes what we can do with the quantum states physically~\cite{Nielsen2010,Wilde2017book,Watrous2011b}. To fully exploit quantum resources, the ability to manipulate quantum channels under operational settings is important. Particularly, in distributed quantum computing, one fundamental primitive, dubbed quantum channel simulation, is to realize quantum channels from one party to another using entanglement and LOCC protocols. Quantum channel simulation, exploiting entanglement to synthesize a target channel through LOCC protocols~\cite{Bennett1996c,Bennett2014,Berta2013,Pirandola2015b,Wilde2018,WW18,Fang2018,Pirandola2018}, servers as the basis of many problems in quantum information, including quantum communication, quantum metrology~\cite{Pirandola2017}, and quantum key distribution~\cite{Pirandola2020}.

One famous example of quantum channel simulation is quantum teleportation (i.e., simulation of the identity channel). As one of the most important quantum information processing protocols~\cite{bennett1993teleporting,Pirandola2015}, quantum teleportation exploits the physical resource of entanglement to realize noiseless quantum channels between different parties and it is an important building block for quantum technologies including distributed quantum computing and quantum networks. Similar to quantum teleportation, quantum channel simulation is a general technique to send an unknown quantum state $\psi$ from a sender to a receiver such that the receiver could obtain $\cN_{A'\to B}(\psi_{A'})$ with the help of a pre-shared entangled state $\rho_{AB}$ and an LOCC protocol $\Pi$. The overall scheme simulates the target channel $\cN$ in the sense that
\begin{align}
    \Pi(\psi_{A'} \otimes \rho_{AB}) = \cN_{A'\to B}(\psi_{A'}), \forall \psi_{A'}.
\end{align}

For some classes of channels such as Pauli channels, the LOCC-based simulation protocols were known \cite{Bennett1996c,Horodecki1999,Pirandola2015b}. However, the LOCC protocols for general quantum channel simulation is hard to design due to the complexity of LOCC. Even for the qubit amplitude damping (AD) channel, the LOCC protocol for simulating this channel in the non-asymptotic regime is still unknown, and its solution would provide a better estimate of its secret key capacity~\cite{Pirandola2020}. Note that the asymptotic simulation of this channel involving infinite dimensions was introduced in \cite{Pirandola2015b}. 

Here, we apply our LOCCNet to explore the simulation of an AD channel $\mathcal{A}$ using its Choi state~\cite{Choi1975} $\rho_{\mathcal{A}} = (I\otimes\mathcal{A}) (\Phi^+)$ as the pre-shared entangled state. {Note AD channel is one of the realistic sources of noise in superconducting quantum processor~\cite{Chirolli2018}.}

{To train the LOCCNet for simulating $\cA$}, we select a set of linearly independent density matrices $S$ as the training set. The loss function for this channel simulation task is then defined as
\begin{align}
L = -\sum_{\psi\in S} F(\cA(\psi), \cB(\psi)),
\end{align}
where $\cB$ is the actual channel simulated by LOCCNet with current parameters and $F(\rho, \sigma) = \text{Tr}\left(\sqrt{\rho^{1/2}\sigma\rho^{1/2}}\right)^2$ gives the fidelity between states $\rho$ and $\sigma$.
With this loss function to be minimized, the parameters in LOCCNet are optimized to maximize the state fidelity between {$\mathcal{A}(\psi)$ and $\mathcal{B}(\psi)$} for all $\psi \in S$.

Once the LOCCNet is {trained to teleport all the basis states in $S$ with near perfect fidelity}, we obtain a protocol for simulating $\cA$. For benchmarking, we randomly generate 1000 pure states and teleport them to Bob. The results are summarized in Fig.~\ref{fig:AD_channel_sim}. Compared with the original teleportation protocol, we could achieve an equivalent performance at low noise level and a better performance at noise level $\gamma > 0.4$. Note that the numerical simulations are conducted on {Paddle Quantum}~\cite{Paddlequantum}. 

\begin{figure}[t]
\centering
\includegraphics[width=0.5\textwidth]{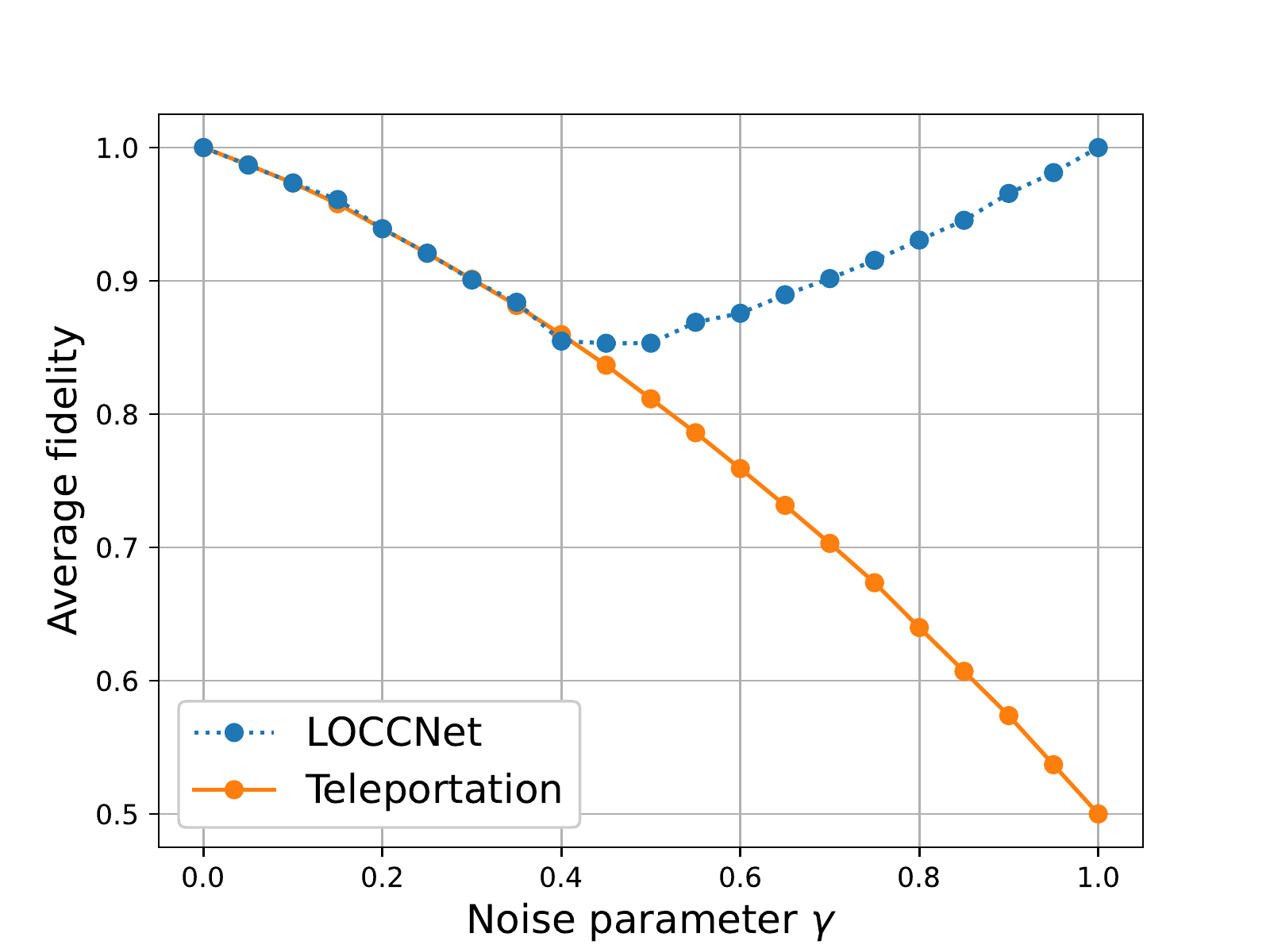}
\caption{\textbf{Average fidelity of simulating AD channel with LOCC protocols.} The blue curve depicts the behavior of the protocol via LOCCNet, which outperforms the original teleportation (orange) at high noise level (noise parameter $\gamma > 0.4$). Each data point contains the statistical results of $1000$ randomly generated states.}
\label{fig:AD_channel_sim}
\end{figure}

\section{Discussion}
We established LOCCNet for exploring LOCC protocols in distributed quantum information processing. {Its overall pipeline is standard for machine learning algorithms. For a specific task, one firstly designs an appropriate loss function and then utilizes different LOCCNet structures and optimization methods to train the model to obtain an optimal or near-optimal protocol. Depending on the nature of the task, a selected training data set may be required, as in the case of channel simulation. Based on the current design of LOCCNet, more machine learning techniques, such as reinforcement learning could be incorporated into this framework, making it a more powerful tool for exploring LOCC protocols.}

LOCCNet not only unifies and extends the existing LOCC protocols, but also sheds light on the power and limitation of LOCC in the noisy intermediate-scale quantum (NISQ) era~\cite{preskill2018quantum} by providing a plethora of examples. We developed improved protocols for entanglement distillation, local state discrimination, and quantum channel simulation as applications. As a showcase, we applied LOCCNet to establish hardware-efficient and state-of-the-art protocols for entanglement distillation of noisy entangled states of interest. In addition to making a significant contribution to entanglement distillation, LOCCNet finds direct practical use in many settings, as we exemplified with several explicit applications in distinguishing noisy and noiseless Bell states as well as simulating amplitude damping channels.

As we have shown the ability of LOCCNet in discovering improved LOCC protocols, one future direction is to apply LOCCNet to further enhance practical entanglement manipulation and quantum communication and explore fundamental problems in quantum information theory. {While in this paper we mainly focus on bipartite cases, LOCCNet also supports multipartite entanglement manipulation. For example, as an essential part in quantum repeaters~\cite{Briegel1998}, entanglement swapping aims to transform two entangled pairs shared between Alice and Bob and between Bob and Carol into a new entangled pair shared by Alice and Carol using only LOCC. Indeed, we could use LOCCNet to design such a protocol. For instance, we can build an LOCCNet where Bob first operates on and measure his subsystem, and then Alice and Carol perform local operations according to the measurement results from Bob. The loss function to minimize can be defined as the infidelity of a target state and the output state shared between Alice and Carol. Similar procedures can be followed to apply LOCCNet in optimizing other multipartite protocols as well, which is worth exploring in future works.}

Another important direction is to extend the framework to the continuous-variable quantum information processing, which may be applied to {explore better} LOCC protocols of private communication based on continuous variable systems~\cite{Pirandola2020}. As we have seen the potential of advancing distributed quantum information processing with the aid of machine learning, we expect more of such cases with classical machine learning being used to improve quantum technologies, which in turn will enhance quantum machine learning applications.

\section*{Data Availability}
Data that support the plots and other findings of this study are available from the corresponding authors upon reasonable request.

\section*{Code Availability}
Code used in the numerical experiments on quantum channel simulation is available at \url{https://github.com/vsastray/LOCCNetcodes}. Other Code used in this study is available from the corresponding authors upon reasonable request.

\section*{Acknowledgements}
We would like to thank Runyao Duan and Kun Fang for helpful discussions.

\section*{Competing Interests}
The authors declare no competing interests.

\section*{Author Contributions}
X. W. formulated the initial idea and the framework; X. Z. and B. Z. developed the theory; X. Z., B. Z., Z. W., and Z. S. performed the experiments.
All co-authors contributed to the preparation of the manuscript.


\newpage
\vspace{2cm}
\onecolumngrid
\vspace{2cm}
\begin{center}
{\textbf{\large Supplemental Information: Practical distributed quantum information processing with LOCCNet}}
\end{center}

\renewcommand{\theequation}{S\arabic{equation}}
\renewcommand{\theproposition}{S\arabic{proposition}}
\renewcommand{\thefigure}{S\arabic{figure}}
\renewcommand{\theHequation}{Supplement.\theequation}
\renewcommand{\theHfigure}{Supplement.\thefigure}
\setcounter{equation}{0}
\setcounter{table}{0}
\setcounter{section}{0}
\setcounter{proposition}{0}
\setcounter{figure}{0}

\section*{SUPPLEMENTARY NOTE 1: Details of LOCC}\label{app:xxx}

\textbf{Preliminaries.}
We begin with the preliminaries on quantum information. 
We will frequently use symbols such as $A$ (or $A'$) and $B$ (or $B'$) to denote finite-dimensional Hilbert spaces associated with Alice and Bob, respectively. 
We use $d_A$ to denote the dimension of system $A$. The set of linear operators acting on $A$ is denoted by $\cL(A)$. We usually write an operator with a subscript indicating the system that the operator acts on, such as $M_{AB}$, and write $M_A:=\tr_B M_{AB}$.

A quantum state on system $A$ is a positive operator $\rho_A$ with unit trace. The set of quantum states is denoted as $S(A):= \{\,\rho_A \geq 0 \,|\, \tr \rho_A = 1\,\}$. 
We call a positive operator separable if it can be written as a convex combination of tensor product positive operators. A bipartite positive semidefinite operator $E_{AB} \in \cL(A\ox B)$ is
said to be Positive-Partial-Transpose (PPT) if $E_{AB}^{T_{B}}$ is positive semidefinite. Note that the action of partial transpose (with respect to $B$) is defined as $(\ket{i_A}\bra{k_A}\otimes\ket{j_B}\bra{l_B})^{T_{B}}=\ket{i_A}\bra{k_A}\otimes\ket{l_B}\bra{j_B}$.

\textbf{LOCC.}
When a quantum system is distributed to spatially separated parties, it is natural to consider how the system evolves when the parties perform local quantum operations with classical communication. A systematic definition of LOCC can be found in~\cite{Chitambar2014}. Here, for self-consistency, we give a detailed description of LOCC as follows.

Consider a setting involving multiple spatially separated parties sharing a multipartite quantum system. The set $\text{LOCC}_1$ consists of the most elementary LOCC operations corresponding to LOCC protocols with one classical communication round, where one party performs a local operation and sends the measurement outcome to others, who then perform corresponding local operations on their local systems upon receiving the outcome. A local operation can be described as a set of completely positive (CP) maps $\{\mathcal{E}_m\}$ such that $\sum_m\mathcal{E}_m$ is trace-preserving. The subscript $m$ corresponds to an operation's measurement outcome, which could affect each party's choices of subsequent local operations. A more complicated LOCC operation can be seen as a sequence of $\text{LOCC}_1$ operations. Specifically, for any $r\geq2$, $\text{LOCC}_r$ is defined to be a set of LOCC operations, in which each operation is constructed from an $\text{LOCC}_{r-1}$ operation followed by an $\text{LOCC}_1$ operation. A common characteristic of these LOCC operations is that they can implemented with finite rounds of classical communication. Thus, we define a set $\text{LOCC}_\mathbb{N}$, corresponding to finite round protocols, such that an LOCC operation is in this set if it belongs to $\text{LOCC}_r$ for some $r$ in $\mathbb{N}=\{1,2,\dots\}$.
As there are finite round protocols, there also exist infinite round protocols in theory. These infinite round protocols, together with operations in $\text{LOCC}_\mathbb{N}$, form the set known as $\text{LOCC}$.

LOCCNet is a machine learning framework developed for designing and exploring LOCC protocols for various quantum information processing tasks. In the main text , we give a brief introduction to this framework. Here, we give some common types of LOCC protocols involving two parties, Alice and Bob, as examples to explain how a protocol can be constructed and optimized using the LOCCNet.

\textit{Optimizing one-round LOCC protocols.}
One-round LOCC protocols are protocols having only one round of classical communication. An example is shown in Fig.~\ref{sfig:one_round_circuit}. An application of such a protocol is quantum state teleportation. To optimize a one-round protocol with LOCCNet, we need to build and train three PQCs, shown as a tree in Fig.~\ref{sfig:one_round_tree}. The PQC $U(\theta_0)$ is used to optimize Alice's local operation $U$, and PQCs $V_0(\theta_1)$ and $V_0(\theta_2)$ are for Bob's local operation in the case of Alice measuring $0$ and $1$, respectively.

\begin{figure}[ht]
\centerline{
\Qcircuit @C=1em @R=1.7em {
\lstick{A} & {/} \qw & \gate{U} & \meter \cwx[1] & \qw \\
\lstick{B} & {/} \qw & \qw & \gate{V} & \qw
}
}
\caption{A circuit illustration of one-round LOCC. Alice first performs a local operation and sends the measurement outcome to Bob. Bob then perform a local operation accordingly.}
\label{sfig:one_round_circuit}
\end{figure}
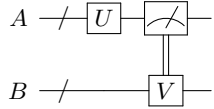

\begin{figure}[h]
\includegraphics[width=0.4\textwidth]{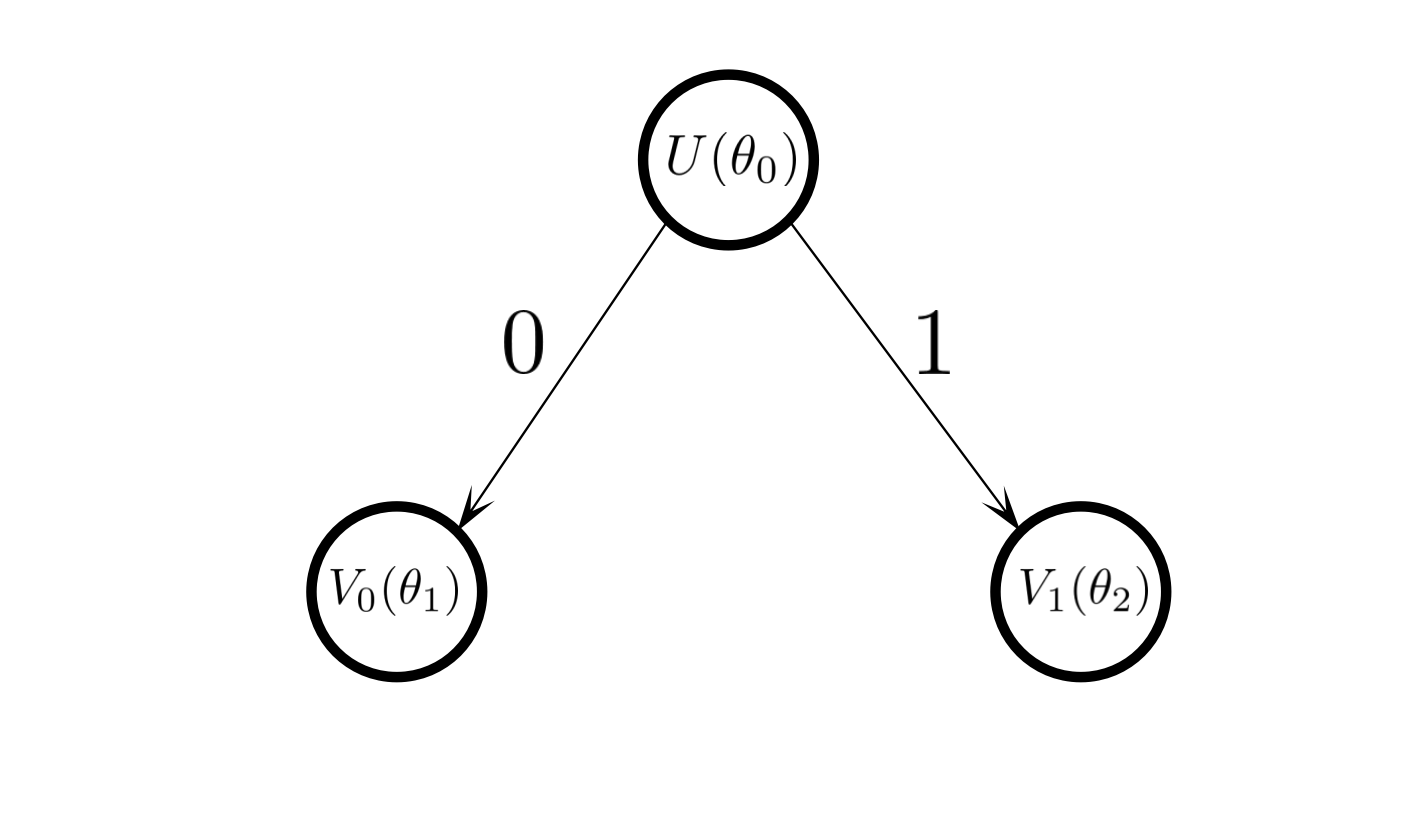}
\caption{Tree structure of the LOCCNet used for optimizing a one-round protocol.}
\label{sfig:one_round_tree}
\end{figure}

\textit{Optimizing two-round LOCC protocols.}
A general two-round LOCC protocol includes Alice performing a local operation and telling Bob her measurement outcome, then Bob performing a corresponding local operation and telling Alice his measurement outcome, and finally Alice performing another local operation. Such a protocol is already a little complicated and optimizing such a protocol requires seven PQCs. Here, we give two special types of two-round protocols that are easier to train and has practical applications.

The first type of protocols is shown in Fig.~\ref{sfig:two_round_protocol_agree} and are widely used for entanglement distillation. In such a protocol, Alice and Bob first perform local operations independently and then exchange their measurement outcomes through classical communication to check whether the expected task is completed. To optimize such a protocol, we only need to build two PQCs, one for Alice's local operation and one for Bob's local operation.
\begin{figure}[!hbtp]
\centerline{
\Qcircuit @C=1em @R=1.7em {
\lstick{A} & {/} \qw & \gate{U} & \meter & \rstick{0\text{ or }1}\cw \\
\lstick{B} & {/} \qw & \gate{V} & \meter & \rstick{0\text{ or }1}\cw
}
}
\caption{A circuit illustration of a type of two-round LOCC where Alice and BOB perform local operations independently before exchanging measurement outcomes.}
\label{sfig:two_round_protocol_agree}
\end{figure}
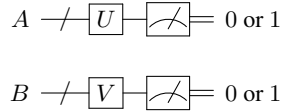

Another type of protocols is given in Fig.~\ref{sfig:two_round_protocol_ctrl}. In such a protocol, after Bob obtains his measurement outcome and tells it to Alice, Alice does not need to perform a local operation. An application of such a protocol is state discrimination, as we show in the main text. Like training a one-round protocol, optimizing a protocol of this type only requires three PQCs.

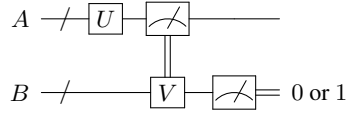
\begin{figure}[!hbtp]
\centerline{
\Qcircuit @C=1em @R=1.7em {
\lstick{A} & {/} \qw & \gate{U} & \meter \cwx[1] & \qw & \qw \\
\lstick{B} & {/} \qw & \qw & \gate{V} & \meter & \rstick{0\text{ or }1}\cw
}
}
\caption{A circuit illustration of another type of two-round LOCC. In such a protocol, Bob sending his measurement outcome to Alice is the last step.}
\label{sfig:two_round_protocol_ctrl}
\end{figure}

\section*{SUPPLEMENTARY NOTE 2: Analysis of entanglement distillation}
The aim of entanglement distillation is to compensate for the impurity caused by noise and restore a maximally entangled state at the cost of many noisy entangled states. In this sense, one could also refer an entanglement distillation protocol as a purification or error-correction protocol.
The Bell states are four two-qubit maximally entangled states defined as
\begin{align}
|\Phi^\pm\rangle = \frac{1}{\sqrt{2}}(|00\rangle\pm|11\rangle),\quad
|\Psi^\pm\rangle = \frac{1}{\sqrt{2}}(|01\rangle\pm|10\rangle).
\end{align}
The state $\ket{\Phi}$ is also known as the entangled bit (ebit), and entanglement distillation in two-qubit settings usually means to convert copies of a state $\rho_{AB}$ shared by two parties, Alice and Bob, into a state closer to the ebit. Here, closeness between the state $\rho_{AB}$ and the ebit is usually measured in terms of the fidelity
\begin{align}
F = \langle\Phi^+|\rho_{AB}|\Phi^+\rangle.
\end{align}

A well known protocol for two-copy entanglement distillation is the DEJMPS protocol, which is illustrated in Fig.~\ref{sfig:DEJMPS_protocol}. Sharing two copies of an initial state, $\rho_{A_0B_0}$ and $\rho_{A_1B_1}$, both Alice and Bob first apply $R_x$ gates and CNOT gates to their local qubits and then measure a pair of qubits from the same copy. Finally, they exchange measurement outcomes and output the unmeasured copy when their outcomes agree. Otherwise, the distillation procedure fails.

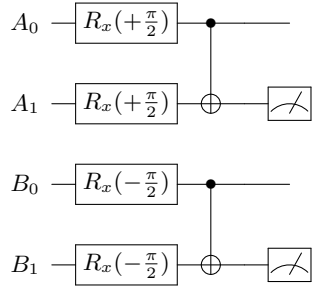
\begin{figure}[!hbtp]
\centerline{
\Qcircuit @C=1em @R=1.7em {
\lstick{A_0}  &\gate{R_x(+\frac{\pi}{2})}& \ctrl{1} &\qw &\qw \\
\lstick{A_1}  &\gate{R_x(+\frac{\pi}{2})}& \targ    &\qw &\meter \\
\lstick{B_0}  &\gate{R_x(-\frac{\pi}{2})}& \ctrl{1} &\qw &\qw \\
\lstick{B_1}  &\gate{R_x(-\frac{\pi}{2})}& \targ    &\qw &\meter \\
}
}
\caption{The DEJMPS protocol for two-copy entanglement distillation.}
\label{sfig:DEJMPS_protocol}
\end{figure}

The DEJMPS protocol has been shown to be optimal in purifying two copies of any Bell diagonal state with of rank at most three~\cite{rozpkedek2018optimizing}, where a Bell diagonal state is a state of the form
\begin{align}
\rho_{AB} = p_0\proj{\Phi^+} + p_1\proj{\Psi^+} + p_2\proj{\Phi^-} + p_3\proj{\Psi^-},
\end{align}
which is a convex combination of the four Bell states. For conciseness, we can write such a Bell diagonal state as a $4$-tuple,
\begin{align}
\rho_{AB} = (p_0, p_1, p_2, p_3).
\end{align}
The DEJMPS protocol can also distill some states besides Bell diagonal states, like S states. In the following, we will analyze the performance of the DEJMPS protocol on two copies of an S state and compare it with a protocol learned by LOCCNet. After that, we will compare the DEJMPS protocol with another protocol learned by LOCCNet for distilling four copies of an isotropic state, which is a special Bell diagonal state.

\textbf{S state.}\label{ssec:s_state}
The S state is defined as the Bell state with a non-orthogonal product noise,
\begin{align}\label{seq:s_state}
\rho_{AB}(p) = p|\Phi^+\rangle\langle\Phi^+| + (1-p)|00\rangle\langle00|,
\end{align}
where $p\in[0,1]$. In the main text, we give expressions of fidelity achieved by the DEJMPS protocol and the protocol learned by LOCCNet for two copies of an S state. Here, we give a detailed derivation of these two expressions.
\begin{proposition}
For two copies of an S state with parameter $p$, the DEJMPS protocol outputs a state whose fidelity to the ebit is
\begin{align}
F = \frac{(1+p)^2}{2+2p^2}
\end{align}
with a probability of success
\begin{align}
p_\text{succ} = \frac{1+p^2}{2}.
\end{align}
\end{proposition}
\begin{proof}
By its definition in Equation~\eqref{seq:s_state}, an S state $\rho$ with parameter $p$ can be written in the matrix form as
\begin{align}\label{seq:s_state_matrix}
\rho=
\begin{pmatrix}
1-\frac{p}{2} & 0 & 0 & \frac{p}{2} \\
0 & 0 & 0 & 0 \\
0 & 0 & 0 & 0 \\
\frac{p}{2} & 0 & 0 & \frac{p}{2} \\
\end{pmatrix}.
\end{align}
Applying the circuit in Fig.~\ref{sfig:DEJMPS_protocol} to two copies of such an state, Alice and Bob both get $0$ for measurement outcomes with a probability of $p_{00} = (1+p^2)/4$. By matrix calculation, we obtain the post-measurement state of the unmeasured copy in this case as
\begin{align}
\sigma_-=
\begin{pmatrix}
\alpha & -\beta & -\beta & \alpha \\
-\beta & \beta & \beta & -\beta \\
-\beta & \beta & \beta & -\beta \\
\alpha & -\beta & -\beta & \alpha \\
\end{pmatrix},
\end{align}
where $\alpha = (1+p)^2/(4+4p^2)$ and $\beta = (1-p)^2/(4+4p^2)$. The probability that Alice and Bob both get $1$ for measurement outcomes is $p_{11} = (1+p^2)/4$, and the post-measurement state in this case is 
\begin{align}
\sigma_+=
\begin{pmatrix}
\alpha & \beta & \beta & \alpha \\
\beta & \beta & \beta & \beta \\
\beta & \beta & \beta & \beta \\
\alpha & \beta & \beta & \alpha \\
\end{pmatrix}.
\end{align}
According to the definition of fidelity, the fidelity of state $\sigma_\pm$ to the ebit is
\begin{align}
F &= \text{Tr}(\sigma_\pm|\Phi^+\rangle\langle\Phi^+|)
 = \frac{(1+p)^2}{2+2p^2}.
\end{align}
The probability of Alice and Bob arriving at state $\sigma_\pm$ is
\begin{align}
p_\text{succ} = p_{00} + p_{11} = \frac{1+p^2}{4} + \frac{1+p^2}{4} = \frac{1+p^2}{2}.
\end{align}
\end{proof}

With LOCCNet, we are able to find a new protocol that achieves a higher fidelity than the DEJMPS protocol when distilling two copies of an S state. Indeed, we show in the main text that this protocol is optimal in the sense that it achieves the highest possible fidelity. With some simplification, we obtain a circuit shown in Fig.~\ref{sfig:s_state_distill_circuit}. Below, we offer analysis on the performance of this optimized protocol in Proposition~\ref{sprop:s_state_fid}.

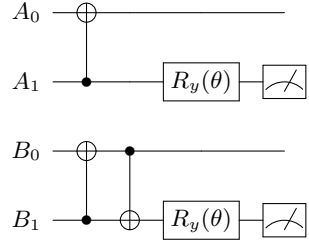
\begin{figure}[!hbtp]
\centerline{
\Qcircuit @C=1em @R=1.7em {
\lstick{A_0} & \targ &\qw &\qw &\qw \\
\lstick{A_1} & \ctrl{-1} &\qw & \gate{R_y(\theta)} &\meter \\
\lstick{B_0} & \targ & \ctrl{1} &\qw &\qw \\
\lstick{B_1} & \ctrl{-1} & \targ & \gate{R_y(\theta)} &\meter \\
}
}
\caption{The simplified circuit of a distillation protocol learned by LOCCNet for two copies of some S state, $\rho_{A_0B_0}$ and $\rho_{A_1B_1}$. The rotation angles of both $R_y$ gates are $\theta = \arccos(1-p)+\pi$, which depends on the parameter $p$ of the S states to be distilled.}
\label{sfig:s_state_distill_circuit}
\end{figure}

\begin{proposition}\label{sprop:s_state_fid}
For two copies of an S state with parameter $p$, the protocol illustrated in Fig.~\ref{sfig:s_state_distill_circuit} outputs a state whose fidelity to the ebit is
\begin{align}
F = \frac{1+\sqrt{2p-p^2}}{2}
\end{align}
with probability $p_\text{succ} = p^2 - \frac{p^3}{2}$ of success.
\end{proposition}
\begin{proof}
The matrix form of an S state $\rho$ with parameter $p$ is given in Equation~\eqref{seq:s_state_matrix}. Applying the circuit in Fig.~\ref{sfig:s_state_distill_circuit} to two copies of such an state, Alice and Bob both get $0$ for measurement outcomes with a probability of $p_{00} = p^2-p^3/2$. By matrix calculation, we obtain the post-measurement state of the unmeasured copy as
\begin{align}
\sigma=
\begin{pmatrix}
\frac{1}{2} & 0 & 0 & \frac{\sqrt{2p-p^2}}{2} \\
0 & 0 & 0 & 0 \\
0 & 0 & 0 & 0 \\
\frac{\sqrt{2p-p^2}}{2} & 0 & 0 & \frac{1}{2} \\
\end{pmatrix}.
\end{align}
Note that the state $\sigma$ can be written as
\begin{align}
\sigma =
\begin{pmatrix}
\alpha+\beta & 0 & 0 & \alpha-\beta \\
0 & 0 & 0 & 0 \\
0 & 0 & 0 & 0 \\
\alpha-\beta & 0 & 0 & \alpha+\beta \\
\end{pmatrix}
= \alpha\proj{\Phi^+} + \beta\proj{\Phi^-},
\end{align}
where $\alpha = (1+\sqrt{2p-p^2})/2$, $\beta = (1-\sqrt{2p-p^2})/2$, and $\ket{\Phi^-} = 1/\sqrt{2}(\ket{00}-\ket{11})$.
By the definition of fidelity, we have
\begin{align}
F &= \text{Tr}(\sigma\proj{\Phi^+})= \text{Tr}((\alpha\proj{\Phi^+} + \beta\proj{\Phi^-})\proj{\Phi^+}) \\
&= \alpha = \frac{1+\sqrt{2p-p^2}}{2}
\end{align}
since $\bra{\Phi^+}\Phi^+\rangle = 1$ and $\bra{\Phi^-}\Phi^+\rangle = 0$ for $\ket{\Phi^-}$ is orthogonal to $\ket{\Phi^+}$. The probability of Alice and Bob arriving at state $\sigma$ is
\begin{align}
p_\text{succ} = p_{00} = p^2 - \frac{p^3}{2}.
\end{align}
\end{proof}

\textbf{Isotropic state.}\label{ssec:iso_state}
A two-qubit isotropic state is of the form
\begin{align}
\rho_{AB} = p|\Phi^+\rangle\langle\Phi^+| + (1-p)\frac{I}{4},
\end{align}
where $p\in[0,1]$. Alternatively, one can write an isotropic state as a Bell diagonal state
\begin{align}
\rho_{AB} = \left(\frac{1+3p}{4}, \frac{1-p}{4}, \frac{1-p}{4}, \frac{1-p}{4}\right).
\end{align}

\textit{Distillation with the DEJMPS protocol.}
The DEJMPS protocol is known to achieve a high fidelity when distilling two copies of an Bell diagonal state, and the resulting state in the event of success is still a Bell diagonal state. Specifically, the DEJMPS protocol's circuit, excluding the measurements, acts on an Bell diagonal state as a permutation of the Bell states' coefficients.
For a Bell diagonal state
\begin{align}
\rho = p_0\proj{\Phi^+} + p_1\proj{\Psi^+} + p_2\proj{\Phi^-} + p_3\proj{\Psi^-},
\end{align}
the operator $R_x(\pi/2)\otimes R_x(-\pi/2)$ maps it to another Bell diagonal state
\begin{align}
\sigma = p_0\proj{\Phi^+} + p_1\proj{\Psi^+} + p_3\proj{\Phi^-} + p_2\proj{\Psi^-}. \label{eq:dejmps_rx}
\end{align}

As stated in Eq.~\eqref{eq:dejmps_rx}, a pair of $R_x(\pm\pi/2)$ gates transforms a Bell diagonal state $(p_0, p_1, p_2, p_3)$ to another Bell diagonal state $(p_0, p_1, p_3, p_2)$. Similarly, a pair of bilateral CNOT gates shown in Fig.~\ref{sfig:DEJMPS_protocol} acts on the tensor product of two Bell diagonal states as a permutation of coefficients. The effect of the bilateral CNOT gates is summarized as a Table in~\cite{Bennett1996}. Specifically, for a pair of Bell diagonal states $(a_0, a_1, a_2, a_3)$ and $(b_0, b_1, b_2, b_3)$, applying the bilateral CNOT gates on the state
\begin{align}
(p_0, p_1,\dots, p_{14}, p_{15}) &= (a_0, a_1, a_2, a_3)\otimes(b_0, b_1, b_2, b_3) \\
&\equiv (a_0b_0, a_0b_1,\dots, a_3b_2, a_3b_3)
\end{align}
results in a state
\begin{align}
\text{CNOT}(p_0, p_1,\dots, p_{14}, p_{15}) = (p_0, p_1, p_{10}, p_{11}, p_5, p_4, p_{15}, p_{14}, p_8, p_9, p_2, p_3, p_{13}, p_{12}, p_7, p_6).
\end{align}

Although the coincidence measurement, referring to Alice and Bob getting identical measurement outcomes, on a Bell diagonal state is not a Bell basis permutation, the post-measurement state is still a Bell diagonal state. To be specific, note that since only $00$ and $11$ are counted as valid results, Bell states $\proj{\Psi^\pm}$ are filtered out and thus a Bell diagonal state $(p_0, p_1, p_2, p_3)$ collapses to $(p_0, 0, p_2, 0)$ up to a normalization factor after the coincidence measurement.

The final fidelity and the probability of success achieved by the DEJMPS protocol can be derived by permuting coefficients in the Bell basis, and the results is given in~\cite{deutsch1996quantum}. For self-consistency, we give a derivation as below.

\begin{proposition}[\cite{deutsch1996quantum}]\label{sprop:dejmps_2copy}
For two copies of a Bell diagonal state $(a_0, a_1, a_2, a_3)$, the DEJMPS protocol outputs a state whose fidelity to the ebit is
\begin{align}
F = \frac{a_0^2+a_3^2}{p_\text{succ}},
\end{align}
where $p_\text{succ} = a_0^2+a_3^2 + a_1^2+a_2^2 + 2a_0a_3 + 2a_1a_2$ is the probability of success.
\end{proposition}
\begin{proof}
After the first layer of $R_x$ gates, the input state $(a_0,a_1,a_2,a_3)^{\otimes2}$ becomes $(a_0,a_1,a_3,a_2)^{\otimes2}$ according to Eq.~\eqref{eq:dejmps_rx}. Then, transformed by the layer of bilateral CNOT gates, the state $(p_0, p_1,\dots, p_{14}, p_{15}) = (a_0,a_1,a_3,a_2)^{\otimes2}$ becomes $(p_0, p_1, p_{10}, p_{11}, p_5, p_4, p_{15}, p_{14}, p_8, p_9, p_2, p_3, p_{13}, p_{12}, p_7, p_6)$. The coincidence measurement in the computational basis on the second copy filters out $\proj{\Psi^\pm}$ and the remaining state is either $\proj{00}$ or $\proj{11}$. In either case, the first copy becomes
\begin{align}
\sigma &= (p_0+p_{10}, p_5+p_{15}, p_8+p_2, p_{13}+p_7) \\
&= (a_0^2+a_3^2, a_1^2+a_2^2, a_3a_0+a_0a_3, a_2a_1+a_1a_2) \\
&= (a_0^2+a_3^2, a_1^2+a_2^2, 2a_0a_3, 2a_1a_2)
\end{align}
up to a normalization factor. The sum of all the unnormalized coefficients is the probability of measuring $00$ or $11$, i.e.,
\begin{align}
p_\text{succ} = a_0^2+a_3^2 + a_1^2+a_2^2 + 2a_0a_3 + 2a_1a_2.
\end{align}
The the normalized output state is $\sigma/p_\text{succ}$, and its fidelity to the ebit, which is the coefficient before $\proj{\Phi^+}$, is
\begin{align}
F = \frac{a_0^2+a_3^2}{p_\text{succ}}.
\end{align}
\end{proof}

Since the DEJMPS protocol is for distilling two copies of a Bell diagonal state, to distill four copies of an isotropic state, we can follow these steps. First, we divide them into two groups where each group consists of two copies. Then we apply the DEJMPS protocol to both groups independently. In the event of success, we will get two copies of a Bell diagonal state, to which we apply the DEJMPS protocol again.

\begin{proposition}
For four copies of an isotropic state $((1+3p)/4, (1-p)/4, (1-p)/4, (1-p)/4)$, the generalized DEJMPS protocol given above outputs a state whose fidelity to the ebit is
\begin{align}
F = \frac{1+10p^2+8p^3+13p^4}{4+8p^2+20p^4}
\end{align}
with $p_\text{succ} = \frac{1}{8}\left(1+2p^2+5p^4\right)$ probability of success.
\end{proposition}
\begin{proof}
For two copies of the isotropic state to be distilled, the DEJMPS protocol outputs a state
\begin{align}
\rho &= \left(\frac{(1+3p)^2+(1-p)^2}{16p'_\text{succ}}, \frac{(1-p)^2+(1-p)^2}{16p'_\text{succ}}, \frac{2(1+3p)(1-p)}{16p'_\text{succ}}, \frac{2(1-p)(1-p)}{16p'_\text{succ}}\right) \\
&= \left(\frac{1+2p+5p^2}{8p'_\text{succ}}, \frac{1-2p+p^2}{8p'_\text{succ}}, \frac{1+2p-3p^2}{8p'_\text{succ}}, \frac{1-2p+p^2}{8p'_\text{succ}}\right)
\end{align}
up to a normalization factor with a probability of success
\begin{align}
p'_\text{succ} &= \frac{1}{16}\left((1+3p)^2 + 3(1-p)^2 + 2(1+3p)(1-p) + 2(1-p)^2\right) 
= \frac{1+p^2}{2}.
\end{align}

Then, the probability of successful distillation for both groups is $p_\text{succ}^{'2}$. In that case, applying the DEJMPS protocol to the resulting two copies of state $\rho$ gives a state whose fidelity to the ebit is
\begin{align}
F &= \frac{(1+2p+5p^2)^2 + (1-2p+p^2)^2}{64p_\text{succ}^{'2}p''_\text{succ}} = \frac{1+10p^2+8p^3+13p^4}{8(1+p^2)^2p''_\text{succ}}
\end{align}
with a probability of success
\begin{align}
p''_\text{succ} = \frac{1+2p^2+5p^4}{2(1+p^2)^2}.
\end{align}
Substituting $p''_\text{succ}$ into $F$, we have
\begin{align}
F = \frac{1+10p^2+8p^3+13p^4}{4+8p^2+20p^4}.
\end{align}
The sucess probability of the whole process is
\begin{align}
p_\text{succ} = p_\text{succ}^{'2}p''_\text{succ} = \frac{1}{8}\left(1+2p^2+5p^4\right).
\end{align}
\end{proof}

\textit{The protocol found with LOCCNet.}
As we show in the main text, the DEJMPS protocol does not fully exploit the resources encoded in four copies of an isotropic states, and there is a protocol learned by LOCCNet that achieves a higher fidelity.
\begin{figure}[!hbtp]
\centerline{
\Qcircuit @C=1em @R=1.7em {
\lstick{A_0} & \ctrl{1} & \qw & \qw & \targ & \qw &\qw \\
\lstick{A_1} & \targ & \ctrl{1} & \qw & \qw & \gate{R_x(+\frac{\pi}{2})} & \meter \\
\lstick{A_2} & \qw & \targ & \ctrl{1} & \qw & \gate{R_x(+\frac{\pi}{2})} & \meter \\
\lstick{A_3} & \qw & \qw & \targ & \ctrl{-3} & \gate{R_x(+\frac{\pi}{2})} & \meter \\
}
}
\caption{The simplified circuit of a protocol learned by LOCCNet for entanglement distillation with four copies of some isotropic state. This circuit only includes Alice's operation, while Bob's operation is identical to Alice's, except that the rotation angles of Bob's $R_x$ gates are $-\pi/2$.}
\label{sfig:4copy_iso_circuit}
\end{figure}

\begin{proposition}
For four copies of an isotropic state $\rho$ with parameter $p$, the protocol illustrated in Fig.~\ref{sfig:4copy_iso_circuit} outputs a state whose fidelity to the ebit is
\begin{align}
F = \frac{1-2p+9p^2}{4-8p+12p^2}
\end{align}
with a probability of success
\begin{align}
p_\text{succ} = \frac{1+4p^3+3p^4}{8}.
\end{align}
\end{proposition}
\begin{proof}
Similar to the DEJMPS protocol, this optimized protocol consists of $Ry(\pm\pi/2)$ gates, bilateral CNOT gates, and coincidence measurements in the computational basis. Thus, the claimed fidelity and probability of success can be derived by simulating the circuit shown in Fig.~\ref{sfig:4copy_iso_circuit} as permutation on Bell basis. Using similar techniques from the proof of Proposition~\ref{sprop:dejmps_2copy}, we obtain the unnormalized state after three coincidence measurements in the event of success, which is
\begin{align}
\sigma = \left(\frac{1}{32}(1+p)^2(1-2p+9p^2), \frac{1}{32}(1-p^2)^2, \frac{1}{32}(1-p^2)^2, \frac{1}{32}(1-p^2)^2\right).
\end{align}
Then, adding up all the coefficients in $\sigma$, we obtain the probability of success
\begin{align}
p_\text{succ} &= \frac{1}{32}(1+p)^2(1-2p+9p^2) + \frac{3}{32}(1-p^2)^2 = \frac{1+4p^3+3p^4}{8}.
\end{align}
Meanwhile, the normalized $\sigma$'s fidelity to the ebit is
\begin{align}
F &= \frac{1}{32p_\text{succ}}(1+p)^2(1-2p+9p^2) = \frac{1-2p+9p^2}{4-8p+12p^2}.
\end{align}
\end{proof}

\textbf{PPT bound.}\label{ssec:distillation_ppt}
As the mathematical structure of LOCC is complex and difficult to characterize~\cite{Chitambar2014}, we may consider larger but mathematically more tractable classes of operations.  
The operations most frequently employed beyond
LOCC are the PPT operations, which completely preserve the positivity of the partial transpose~\cite{Rains2001}. 
A bipartite quantum operation $\Pi_{AB \to A'B'}$ is called a PPT operation if its  Choi-Jamio\l{}kowski matrix $J_\Pi = \sum_{i,j,m,k} \ket{i_Aj_B}\bra{m_Ak_B} \ox \Pi(\ket{i_Aj_B}\bra{m_Ak_B})$ is positive under partial transpose  across the bipartition of $AA':BB'$, where $\{\ket{i_A}\}$ and $\{\ket{j_B}\}$ are orthonormal bases on Hilbert spaces $A$ and $B$, respectively. 

The entanglement theory under PPT operations has been extensively studied in the literature~(e.g.,~\cite{Audenaert2003,Wang2016d,Matthews2008,Wang2020c,Regula2019,Chitambar2017}) and offers the limitations of LOCC. In particular, the limit of finite-copy entanglement distillation was recently explored in \cite{Fang2017,rozpkedek2018optimizing,Regula2019}. In the following, we compare our results with the PPT bound from \cite{rozpkedek2018optimizing}, which gives the fundamental limits on the fidelity of distillation with given success probability.
To be specific, the maximal fidelity of distilling $D$-dimensional Bell state from $\rho$ with fixed success probability $\delta$ using PPT operations \cite{rozpkedek2018optimizing} is given by
\begin{equation}
\begin{array}{ll}
\operatorname{maximize} & \frac{d_Ad_B}{\delta} \tr \rho_{AB}^{T} M_{AB} \\
\text { subject to } & M_{AB} \geqslant 0, \quad E_{AB} \geqslant 0, \\
& M_{AB}+E_{AB} \leqslant \frac{\mathbb{I}_{AB}}{d_Ad_B}, \\
& M_{AB}^{T_B}+E_{AB}^{T_B} \leqslant \frac{\mathbb{I}_{AB}}{d_Ad_B},  d_Ad_B \tr \left[\rho_{AB}^{T}\left(M_{A B}+E_{AB}\right)\right]=\delta, \\
& M_{AB}^{T_B}+\frac{1}{D+1} E_{AB}^{T_B} \geqslant 0,  -M_{AB}^{T_B}+\frac{1}{D-1} E_{AB}^{T_B} \geqslant 0,
\end{array}
\end{equation}
where $d_A, d_B$ are the dimensions of systems $A$ and $B$, respectively.
Recall that $\rho_{AB}$ is the initial input state that Alice and Bob are attempting to distill and in most examples considered here, 
it will consist of two copies of some two-qubit state. 
\section*{SUPPLEMENTARY NOTE 3: Analysis of LOCC state discrimination}
To explore the power of LOCCNet in state discrimination, we focus on the optimal success probability of discriminating noiseless and noisy Bell states via LOCC. In the following, we present the LOCC protocol from \cite{Walgate2000} for discriminating two Bells states. After that, we show how to distinguish one Bell state from one noisy Bell state using the protocol learned via LOCCNet and compare it with the protocol of the noiseless case.

\textbf{Noiseless case.}\label{ssec:noiseless_bell}
Consider two Bell states, $|\Phi^+\rangle$ and $|\Phi^-\rangle$. Since these two states are pure and orthogonal to each other, there exist an LOCC protocol that can perfectly distinguish between them \cite{Walgate2000}. Here, we give a specific discrimination protocol for these two Bell states. Suppose Alice and Bob share a two-qubit state $\rho_{AB}$, which could be either $\proj{\Phi^+}$ or $\proj{\Phi^-}$. To find out which state it is through LOCC, they can follow the steps below. First, Alice applies a $R_y(\pi/2)$ gate on her qubit followed by a measurement. Then, Alice tell Bob her measurement outcome through classical communication. Receiving the measurement outcome from Alice, Bob applies on his qubit a $R_y$ gate with the rotation angle $\theta$ being $\pi/2$ or $-\pi/2$, corresponding to the case where the communicated measurement outcome is $0$ or $1$, respectively. Finally, Bob measures his qubit. If he gets $0$, then he can be sure that the state $\rho_{AB}$ is $\proj{\Phi^+}$. Otherwise, $\rho_{AB} = \proj{\Phi^-}$. 
The whole process is also illustrated with a circuit shown in Fig.~\ref{sfig:qsd_noiseless_protocol}, which can perfectly discriminate $|\Phi^+\rangle$ and $|\Phi^-\rangle$.

\begin{figure}[!hbtp]
\centerline{
\Qcircuit @C=1em @R=1.7em {
\lstick{A}  & \gate{R_y(\frac{\pi}{2})} & \meter \cwx[1] & \qw & \qw \\
\lstick{B}  & \qw & \gate{R_y(\theta)}& \meter & \qw \\
}
}
\caption{The LOCC protocol distinguishing between the pair of noiseless Bell states $|\Phi^+\rangle$ and $|\Phi^-\rangle$. The rotation angle $\theta$ of Bob's $R_y$ gate is either $\pi/2$ or $-\pi/2$, depending on Alice's measurement outcome.}
\label{sfig:qsd_noiseless_protocol}
\end{figure}
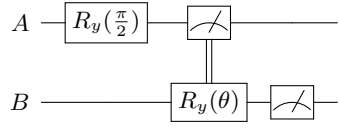

\textbf{Noisy case.}\label{sec:noisy bell}
Quantum noises unavoidably may occur in quantum information processing. One common noise of theoretical and experimental interest is the amplitude damping channel~\cite{Nielsen2010}, which is one of the realistic sources of noise in superconducting quantum processor~\cite{Chirolli2018}. To be specific, an amplitude damping (AD) channel $\mathcal{A}$ with noise parameter $\gamma$ such that $\mathcal{A}(\rho) = E_0\rho E_0^\dagger+E_1\rho E_1^\dagger$ with $E_0=|0\rangle\langle0|+\sqrt{1-\gamma}|1\rangle\langle1|$ and $E_1=\sqrt{\gamma}|0\rangle\langle1|$.
If $|\Phi^-\rangle$ is affected by the amplitude damping noise on each qubit, then the resulting state is $\mathcal{A}\otimes\mathcal{A}(|\Phi^-\rangle\langle\Phi^-|)$.
The goal is now to distinguish between $\Phi_0\equiv|\Phi^+\rangle\langle\Phi^+|$ and $\Phi_1\equiv\mathcal{A}\otimes\mathcal{A}(|\Phi^-\rangle\langle\Phi^-|)$.

\textit{PPT bound.}
The distinguishability of quantum states under PPT (positive partial transpose) POVMs was introduced in \cite{Yu2014a} to better understand the fundamental limits of the local distinguishability of quantum states. To be specific, the PPT POVM used for distinguishing a set of $n$ orthogonal quantum states $\{\rho_1,\dots,\rho_n\}$ can be defined as an $n$-tuple of operators, $(M_k)_{k=1,\dots,n}$, where $M_k$ is PPT for $k=1,\dots,n$ and $\sum_{k=1}^nM_k=I_{AB}$.
The set of PPT POVMs enjoys a more tractable mathematical structure than the LOCC POVMs due to the SDP characterization of PPT condition. 

The optimal success probability of discriminating a collection of quantum states $\{\rho_1,\cdots,\rho_K\}$ using PPT POVMs is given by
\begin{equation}
\begin{split}
p_s(\rho_1,\cdots,\rho_K)  = \max \ & \frac{1}{K}\tr \sum_{j=1}^K\rho_j \\
\text{s.t.} \ & \sum_{j=1}^K M_j = I, 0\le M_j\le I, \forall j=1,2,\cdots,K,\\
& M_j^{T_B}\ge 0, \forall j=1,2,\cdots,K.
\end{split}
\end{equation}
where we assume that each state in this collection has a equal probability of appearance.
As LOCC POVMs is a proper subset of PPT POVMs, the above SDP gives the upper bound to the optimal success probability of discriminating a collection of quantum states.

\textit{Optimized LOCC protocol.}
While the PPT bound serves as an upper bound to the success probability of LOCC discrimination, an optimal LOCC protocol may not necessarily reach the bound. Here, we present an LOCC protocol optimized by LOCCNet that achieves a success probability close to the PPT bound.

The only difference between this optimized protocol and the protocol for noiseless discrimination is that the rotation angle $\theta$ of Bob's $R_y$ gate is not fixed in the noisy case, as shown in Fig.~\ref{sfig:qsd_noisy_protocol}. Specifically, for Alice's measurement outcome being $0$,
\begin{align}
\theta = \pi-\arctan\left(\frac{2-\gamma}{\gamma}\right),
\end{align}
where $\gamma$ is the noise parameter of the AD channel $\mathcal{A}$ that $\proj{\Phi^-}$ goes through. For Alice's measurement outcome being $1$,
\begin{align}
\theta = -\pi+\arctan\left(\frac{2-\gamma}{\gamma}\right).
\end{align}

\begin{figure}[!hbtp]
\centerline{
\Qcircuit @C=1em @R=1.7em {
\lstick{A}  & \gate{R_y(\frac{\pi}{2})} & \meter \cwx[1] & \qw & \qw \\
\lstick{B}  & \qw & \gate{R_y(\theta)}& \meter & \qw \\
}
}
\caption{The optimized LOCC protocol for distinguishing between states $\Phi_0$ and $\Phi_1$. Depending on Alice's measurement outcome, the rotation angle $\theta$ of Bob's $R_y$ gate is either $\pi-\arctan(\alpha)$ or $\arctan(\alpha)-\pi$, where $\alpha=(2-\gamma)/\gamma$ and $\gamma$ is the noise parameter of the AD channel.}
\label{sfig:qsd_noisy_protocol}
\end{figure}
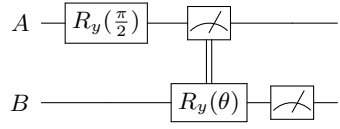

\begin{proposition}\label{sprop:qsd_psucc}
For states $\Phi_0 = \proj{\Phi^+}$ and $\Phi_1 = \mathcal{A}\otimes\mathcal{A}(\proj{\Phi^-})$, the optimized protocol illustrated in Fig.~\ref{sfig:qsd_noisy_protocol} discriminates between them with an average probability of
\begin{align}
p_\text{succ} = \frac{1}{2} + \frac{\sqrt{2-2\gamma+\gamma^2}}{2\sqrt{2}}.
\end{align}
\end{proposition}
\begin{proof}
We consider this protocol in two cases.

\textit{Case 1. The input state is $\Phi_0 = \proj{\Phi^+}$.} For this case, after Alice applies $R_y(\pi/2)$ to her qubit, state $\ket{\Phi^+}$ becomes
\begin{align}
R_y\left(\frac{\pi}{2}\right)\otimes I \ket{\Phi^+} = \frac{1}{2}\ket{0}\otimes(\ket{0}-\ket{1}) + \frac{1}{2}\ket{1}\otimes(\ket{0}+\ket{1}).
\end{align}
Measuring the first qubit of this state, Alice has $50\%$ chance of getting $0$ with the normalized post-measurement state being $\ket{\psi_0} = 1/\sqrt{2}\ket{0}\otimes(\ket{0}-\ket{1})$ and $50\%$ chance of getting $1$ with state $\ket{\psi_1} = 1/\sqrt{2}\ket{1}\otimes(\ket{0}+\ket{1})$. The resulting states of the second qubit after Bob's operation are given below, where $\theta_0 = \pi-\arctan((2-\gamma)/\gamma)$ and $\theta_1 = -\pi+\arctan((2-\gamma)/\gamma)$.
\begin{align}
R_y(\theta_0)\frac{1}{\sqrt{2}}(\ket{0}-\ket{1}) &= 
\frac{1}{\sqrt{2}}\left(\left(\cos\frac{\theta_0}{2} + \sin\frac{\theta_0}{2}\right)\ket{0} - \left(\cos\frac{\theta_0}{2} - \sin\frac{\theta_0}{2}\right)\ket{1}\right); \\
R_y(\theta_1)\frac{1}{\sqrt{2}}(\ket{0}+\ket{1}) 
&= \frac{1}{\sqrt{2}}\left(\left(\cos\frac{\theta_1}{2} - \sin\frac{\theta_1}{2}\right)\ket{0} + \left(\cos\frac{\theta_1}{2} + \sin\frac{\theta_1}{2}\right)\ket{1}\right).
\end{align}
Then, the probability of Bob's measurement outcome being $0$ given $\Phi_0$ as the input state is
\begin{align}
P(0|\Phi_0) &= \frac{1}{2}\left|\frac{1}{\sqrt{2}}\left(\cos\frac{\theta_0}{2} + \sin\frac{\theta_0}{2}\right)\right|^2 + \frac{1}{2}\left|\frac{1}{\sqrt{2}}\left(\cos\frac{\theta_1}{2} - \sin\frac{\theta_1}{2}\right)\right|^2 
= \frac{1}{2} + \frac{\sin\theta_0-\sin\theta_1}{4}.
\end{align}
Since $\theta_1 = -\theta_0$, we have $\sin\theta_1 = -\sin\theta_0$ and
\begin{align}
P(0|\Phi_0) &= \frac{1}{2} + \frac{\sin\theta_0+\sin\theta_0}{4} = \frac{1+\sin\theta_0}{2}.
\end{align}
Let $\alpha\equiv(2-\gamma)/\gamma$. Then, $\theta_0 = \pi-\arctan\alpha$ and
\begin{align}
P(0|\Phi_0) &= \frac{1+\sin(\pi-\arctan\alpha)}{2} 
= \frac{1+\sin(\arctan\alpha)}{2} 
= \frac{1}{2}+\frac{\alpha}{2\sqrt{1+\alpha^2}}.
\end{align}

\textit{Case 2. The input state is $\Phi_1 = \mathcal{A}\otimes\mathcal{A}(\proj{\Phi^-})$.} If the noise parameter of the AD channel $\mathcal{A}$ is $\gamma$, then the input state in matrix form is
\begin{align}
\Phi_0 = \frac{1}{2}
\begin{pmatrix}
1+\gamma^2 & 0 & 0 & \gamma-1 \\
0 & \gamma-\gamma^2 & 0 & 0 \\
0 & 0 & \gamma-\gamma^2 & 0 \\
\gamma-1 & 0 & 0 & (1-\gamma)^2
\end{pmatrix}.
\end{align}
After Alice applies $R_y(\pi/2)$ to her qubit, the input state becomes
\begin{align}
R_y\left(\frac{\pi}{2}\right)\otimes I \Phi_0 R_y^\dagger\left(\frac{\pi}{2}\right)\otimes I =\frac{1}{4}
\begin{pmatrix}
1+\gamma & 1-\gamma & 1-\gamma+2\gamma^2 & \gamma-1 \\
1-\gamma & 1-\gamma & 1-\gamma & -1+3\gamma-2\gamma^2 \\
1-\gamma+2\gamma^2 & 1-\gamma & 1+\gamma & \gamma-1 \\
\gamma-1 & -1+3\gamma-2\gamma^2 & \gamma-1 & 1-\gamma
\end{pmatrix}.
\end{align}
From this state, we can see that when Alice measures her qubit, she has $50\%$ chance of getting $0$ with this state collapsing to
\begin{align}
\rho_0 = \proj{0}\otimes\frac{1}{2}
\begin{pmatrix}
1+\gamma & 1-\gamma \\
1-\gamma & 1-\gamma
\end{pmatrix}.
\end{align}
If Alice gets $1$ instead, this state collapses to
\begin{align}
\rho_1 = \proj{1}\otimes\frac{1}{2}
\begin{pmatrix}
1+\gamma & \gamma-1 \\
\gamma-1 & 1-\gamma
\end{pmatrix}.
\end{align}
The resulting states of the second qubit after Bob's operation are given below, where $\theta_0 = \pi-\arctan((2-\gamma)/\gamma)$ and $\theta_1 = -\pi+\arctan((2-\gamma)/\gamma)$.
\begin{align}
R_y(\theta_0)\frac{1}{2}
\begin{pmatrix}
1+\gamma & 1-\gamma \\
1-\gamma & 1-\gamma
\end{pmatrix}
R_y^\dagger(\theta_0)
&= \frac{1}{2}
\begin{pmatrix}
1+\gamma\cos\theta_0 + (\gamma-1)\sin\theta_0 & (1-\gamma)\cos\theta_0 + \gamma\sin\theta_0 \\
(1-\gamma)\cos\theta_0 + \gamma\sin\theta_0 & 1-\gamma\cos\theta_0 + (1-\gamma)\sin\theta_0
\end{pmatrix}; \\
R_y(\theta_1)\frac{1}{2}
\begin{pmatrix}
1+\gamma & \gamma-1 \\
\gamma-1 & 1-\gamma
\end{pmatrix}
R_y^\dagger(\theta_1)
&= \frac{1}{2}
\begin{pmatrix}
1+\gamma\cos\theta_1 + (1-\gamma)\sin\theta_1 & (\gamma-1)\cos\theta_1 + \gamma\sin\theta_1 \\
(\gamma-1)\cos\theta_1 + \gamma\sin\theta_1 & 1-\gamma\cos\theta_1 + (\gamma-1)\sin\theta_1
\end{pmatrix}.
\end{align}
Then, the probability of Bob's measurement outcome being $1$ given $\Phi_1$ as the input state is
\begin{align}
P(1|\Phi_1) &= \frac{1}{2}\cdot\frac{1-\gamma\cos\theta_0 + (1-\gamma)\sin\theta_0}{2} + \frac{1}{2}\cdot\frac{1-\gamma\cos\theta_1 + (\gamma-1)\sin\theta_1}{2} \\
&= \frac{2 - \gamma(\cos\theta_0+\cos\theta_1) + (1-\gamma)(\sin\theta_0-\sin\theta_1)}{4}.
\end{align}
Since $\theta_0 = \theta_1$, we have $\sin\theta_1=-\sin\theta_0$, $\cos\theta_1 = \cos\theta_0$, and thus
\begin{align}
P(1|\Phi_1) &= \frac{2 - \gamma(\cos\theta_0+\cos\theta_0) + (1-\gamma)(\sin\theta_0+\sin\theta_0)}{4} \\
&= \frac{1 - \gamma\cos\theta_0 + (1-\gamma)\sin\theta_0}{2}.
\end{align}
Let $\alpha\equiv(2-\gamma)/\gamma$. Then $\theta_0 = \pi-\arctan\alpha$ and
\begin{align}
P(1|\Phi_1) 
= \frac{1}{2} + \frac{\gamma+(1-\gamma)\alpha}{2\sqrt{1+\alpha^2}}.
\end{align}

Combining these two cases, we can obtain this optimized protocol's average probability of success as
\begin{align}
p_\text{succ} &= \frac{1}{2}P(0|\Phi_0) + \frac{1}{2}P(1|\Phi_1) = \frac{1}{2}\left(\frac{1}{2}+\frac{\alpha}{2\sqrt{1+\alpha^2}} + \frac{1}{2} + \frac{\gamma+(1-\gamma)\alpha}{2\sqrt{1+\alpha^2}}\right) \\
&= \frac{1}{2}\left(1+\frac{\gamma+(2-\gamma)\alpha}{2\sqrt{1+\alpha^2}}\right) 
= \frac{1}{2}+\frac{\sqrt{2-2\gamma+\gamma^2}}{2\sqrt{2}}.
\end{align}
\end{proof}

\end{document}